\documentclass[a4paper,UKenglish,cleveref, autoref, thm-restate]{lipics-v2021}
\usepackage[utf8]{inputenc}
\usepackage[linesnumbered,noend]{algorithm2e}
\usepackage{tikz}
\usetikzlibrary{shapes}
\usepackage{adjustbox}

\usepackage{bm}
\nolinenumbers

\newcommand{\rf}[3]{r_{#3}(#1,#2)}
\newcommand{\ropti}[2]{r^*(#1,#2)}
\newcommand{\graph}{$G=(V,{\cal A},h,t)$ }

\newcommand{\outneighbours}[1]{$\Gamma^+(#1)$}
\newcommand{\inneighbours}[1]{$\Gamma^-(#1)$}
\newcommand{\exitSink}[3]{$\mathbf{S}_#1($#2$,#3)$}

\newcommand{\arcs}{\hat{{\cal A}}}
\newcommand{\val}{\text{val}}
\newcommand{\Max}{{\cal MAX}}
\newcommand{\Min}{{\cal MIN}}
\SetKwComment{Comment}{/* }{ */}

\title{Polynomial Time Algorithm for ARRIVAL on Tree-like Multigraphs}

\author{David Auger}{DAVID Lab.,UVSQ, Université Paris Saclay,  45 avenue des Etats-Unis,78000,Versailles, France \and \url{https://www.david.uvsq.fr/?profile=auger-david} }{david.auger@uvsq.fr}{https://orcid.org/0000-0003-1886-1901}{}

\author{Pierre Coucheney}{DAVID Lab.,UVSQ, Université Paris Saclay, 45 avenue des Etats-Unis,78000,Versailles, France \and \url{https://www.david.uvsq.fr/?profile=coucheney-pierre-2} }{pierre.coucheney@uvsq.fr}{}{}

\author{Loric Duhazé}{DAVID Lab.,UVSQ, Université Paris Saclay, 45 avenue des Etats-Unis,78000,Versailles, France \and {LISN,Université Paris Saclay}, France \and \url{https://www.lri.fr/membre_en.php?mb=2713} }{loric.duhaze@uvsq.fr}{https://orcid.org/my-orcid?orcid=0000-0002-9898-5631}{This research was partially supported by Labex DigiCosme \newline (project ANR11LABEX0045DIGICOSME) operated by ANR as part of the program " Investissement d'Avenir " Idex ParisSaclay (ANR11IDEX000302).}

\authorrunning{D. Auger, P.Coucheney and L.Duhazé}

\Copyright{David Auger, Pierre Coucheney and Loric Duhazé}


 %TODO mandatory; please add comma-separated list of keywords

\acknowledgements{We want to thank Johanne Cohen for her advises and attentive review.}%optional

\EventEditors{John Q. Open and Joan R. Access}
\EventNoEds{2}
\EventLongTitle{42nd Conference on Very Important Topics (CVIT 2016)}
\EventShortTitle{CVIT 2016}
\EventAcronym{CVIT}
\EventYear{2016}
\EventDate{December 24--27, 2016}
\EventLocation{Little Whinging, United Kingdom}
\EventLogo{}
\SeriesVolume{42}
\ArticleNo{23}

\begin{document}

\maketitle

\begin{abstract}

A \emph{rotor walk} in a directed graph can be thought of as a deterministic version of a Markov Chain, where a pebble moves from vertex to vertex following a simple rule until a terminal vertex, or sink, has been reached. The \emph{ARRIVAL problem}, as defined by Dohrau and al.~\cite{dohrau2017arrival}, consists in determining which sink will be reached. While the walk itself can take an exponential number of steps, this problem belongs to the complexity class NP~$\cap$~co-NP without being known to be in P. 
Several variants have been studied where we add one or two players to the model, defining deterministic analogs of stochastic models (e.g., Markovian decision processes, Stochastic Games) with rotor-routing rules instead of random transitions. The corresponding decision problem addresses the existence of strategies for players that ensures some condition on the reached sink. These problems are known to be $NP$-complete for one player and $PSPACE$-complete for two players.
In this work, we define a class of directed graphs, namely \emph{tree-like multigraphs}, which are multigraphs having the global shape of an undirected tree. We prove that in this class, ARRIVAL and its different variants can be solved in linear time, while the number of steps of rotor walks can still be exponential. To achieve this, we define a notion of \emph{return flow}, which counts the number of times the pebble will bounce back in subtrees of the graph.
\end{abstract}

\begin{figure}[!ht]
    \centering
    \includegraphics[scale=0.7]{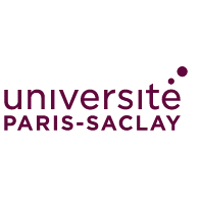}
    
\end{figure}

\newpage

 \section{Introduction}
 %\addcontentsline{toc}{section}{Introduction}
 
 The \emph{rotor routing}, or \emph{rotor walk model}, has been studied under different names: \emph{eulerian walkers} \cite{priezzhev1996eulerian, povolotsky1998dynamics} and \emph{patrolling algorithm} \cite{yanovski2003distributed}. It shares many properties with a more algebraic focused model: \emph{abelian sandpiles} \cite{bjorner1991chip,Holroyd2008}. We can cite  \cite{giacaglia2011local} and \cite{Holroyd2008} as general introductions to this cellular automaton.
 
Let us explain briefly how a rotor walk works. Consider a directed graph and for each vertex $v$, if $v$ has $k$ outgoing arcs, number these arcs from $1$ to $k$. Then, we place a pebble on a starting vertex and proceed to the walk. On the initial vertex, the pebble moves to the next vertex according to arc $1$. It does the same on the second vertex and so on. But the second time that a  vertex is reached, the pebble will move according to arc $2$, and so on until arc $k$ has been used, and then we start again with arc $1$.
 
 This elementary rule defines the rotor routing which has many interesting properties. In particular, it is easily shown that if it the graph is connected enough, the pebble will explore all vertices, but the time needed for such an exploration can be exponential in the number of vertices. 
 
In this work, we fix a set of vertices that we call sinks and stop the walking process when a sink is reached. The problem of determining, for a starting configuration (numbering) of arcs and an initial vertex, which sink will be reached first is the ARRIVAL problem. It is defined in~\cite{dohrau2017arrival} together with a proof that the problem belongs to the complexity class NP~$\cap$~co-NP, but it is not known to be in P. It has then been shown in~\cite{gartner2018arrival} that the problem is in the smaller complexity class UP~$\cap$~co-UP, and a subexponential algorithm has been proposed in~\cite{gartner_et_al:LIPIcs.ICALP.2021.69}, based on computing a Tarski fixed point. This algorithm is even polynomial if the graph is almost acyclic (in a certain sense).
 
 A direct application of the rotor-routing automaton is that several structural properties of Markov chains can be approximated or bounded by rotor walks see \cite{chan2021random,cooper2006simulating,friedrich2010cover}. It seems natural to extend these results to decisional variants and define rotor analogs for  \emph{Markov decision processes} and \emph{stochastic games} \cite{fearnley2017reachability}. It is proved in \cite{fearnley2017reachability} that deciding if a player can ensure some value is NP-complete for the one-player version and PSPACE-complete for the two-player version.

 \paragraph*{Contributions and Organization of the Paper}

In this work, we define the class \emph{tree-like multigraphs}, where the number of steps needed to complete a rotor walk can still be exponential, but where the tree-like structure helps to efficiently solve ARRIVAL, as well as its one and two players variants, in linear time. It is to be noted that tree-like multigraphs are not almost acyclic in the sense  of~\cite{gartner_et_al:LIPIcs.ICALP.2021.69}, thus their algorithm does not run in polynomial time in our case.

In \autoref{sec:defi}, we first give some standard definitions for multigraphs and proceed to define rotor walks in this context together with different rotor-routing notions (\emph{exit pattern, cycle pushing}, etc.).
We begin our study of ARRIVAL by a particular graph, the \emph{Path Graph}, in \autoref{sec:path} and introduce the main tool that we use, namely the \emph{return flow}.
Next, we define tree-like multigraphs and return flow in the general case in \autoref{sec:return}.

The ARRIVAL problem (with zero player) is solved by an almost linear algorithm which is detailed in \autoref{sec:ARRIVALZero}.
Then we proceed to the case of one-player rotor games and two-player rotor games respectively in \autoref{sec:ARRIVAL1P} and \autoref{sec:ARRIVAL2P}.

Finally, in \autoref{sec:ARRIVALSimple} we study and improve the complexity of solving ARRIVAL and the associated decision problems in the particular case of simple graphs \emph{i.e.} graphs that admit a unique arc for each couple of adjacent vertices.

The following table summarizes our results (bold), i.e. time complexity of computing the sink (or the optimal sink, for one-player and two-player) reached by a particle starting on a particular vertex in a graph $G=(V,A)$. The first column states the complexity of the natural algorithm to solve these problems, namely simulating the \emph{rotor walk}.
\begin{center}
\begin{adjustbox}{max totalsize={\textwidth}{\textheight},center}
\begin{tabular}{|c||c|c|c||c|} 
 \hline
 & Rotor Walk & 0 player & 1 player & 2 players  \\ [0.8ex] 
 \hline\hline
 General digraph &  \emph{exponential}&NP $\cap$ coNP & NP-complete & PSPACE-complete  \\[0.5ex]
 
 Tree-like multigraph &  \textbf{exponential} & $\bm{O(|A|)}^\textbf{\dag}$ & $\bm{O(|A|)}$ & $\bm{O(|A|)}$ \\ [0.5ex]
 %\hline
 Simple Tree-like multigraph & $\bm{O(|V|^2)}$ &  $\bm{O(|V|)}^\textbf{\dag}$ & $\bm{O(|V|)}\textbf{*}$ & $\bm{O(|V|)}\textbf{*}$\\ [0.5ex]
 \hline
\end{tabular}
\end{adjustbox}

\end{center}
The \textbf{$\dag$} indicates the cases where we can solve the problem for all vertices of the graph at the same time with this complexity. The \textbf{*} indicates the cases where the problem can be solved for all vertices under certain conditions.

\section{Basic Definitions} \label{sec:defi}

\subsection{Directed Multigraphs}
 
In this paper, unless stated otherwise, we always consider a directed multigraph \graph where $V$ is a finite set of \emph{vertices}, ${\cal A}$ is a finite set of \emph{arcs}, and $h$ (for \emph{head}) and $t$ (for \emph{tail}) are two maps from ${\cal A}$ to $V$ defining incidence between arcs and vertices. For a given arc $a \in {\cal A}$, vertex $h(a)$ is called the head of $a$ and $t(a)$ is the tail of $a$. For sake of clarity, we only consider graphs without arcs of the form $h(a)=t(a)$ (i.e. loops). All our complexity results would remain true if we authorized them. Note that multigraphs can have multiple arcs with the same head and tail.  Let $u \in V$ be a vertex, we denote by ${\cal A}^+(u)$ (resp. ${\cal A}^-(u)$) the subset of arcs $a \in {\cal A}$ with tail $u$ (resp. with head $u$).

Let \outneighbours{u} (resp.~\inneighbours{u}) be the subset of vertices $v \in V$ such that there is an arc $a \in {\cal A}$ with $h(a)=v$ and $t(a)=u$ (resp.$h(a)=u$ and $t(a)=v$). A graph such that for all $u \in V$ we have $|{\cal A}^+(u)|=|\Gamma^+(u)|$ is called \emph{simple}.
A vertex $u$ for which $|\Gamma^+(u) \cup \Gamma^-(u) |=1$ is called a \emph{leaf}.

\subsection{Rotor Routing Mechanics}

\subsubsection*{Rotor Graphs}
  
Let \graph be a multigraph.

\begin{definition}[Rotor Order]
We define a \emph{rotor order} at $u \in V$ as an operator denoted by $\theta_u$ such that:
\begin{itemize}
    \item $\theta_u : {\cal A}^+(u) \rightarrow {\cal A}^+(u)$
    \item for all $a \in {\cal A}^+(u)$, the orbit $\lbrace a, \theta_u(a),\theta_u^2(a),...,\theta_u^{|{\cal A}^+(u)|-1}(a) \rbrace$ of arc $a$ under $\theta_u$ is equal to ${\cal A}^+(u)$, where $\theta_u^k(a)$ is the composition of $\theta_u$ applied to arc $a$ exactly $k$ times.
\end{itemize}
\end{definition}

Observe that each arc of ${\cal A}^+(u)$ appears exactly once in any orbit of $\theta_u$. Now, we will integrate the operator $\theta_u$ to our graph structure as follows.

\begin{definition}[Rotor Graph]
A \emph{rotor graph} $G$ consists in a (multi)graph \graph together with:
\begin{itemize}
    \item a partition $V = V_0 \cup S_0$ of vertices, where $S_0 \neq \emptyset$ is a particular set of vertices called sinks, and $V_0$ is the rest of the vertices;
    \item  a rotor order $\theta_u$ at each $u \in V_0$.
\end{itemize}
\end{definition}

In this document, unless stated otherwise, all the graphs we consider are rotor graphs with $G=(V_0,S_0,{\cal A},h,t,\theta)$.

\begin{definition}[Rotor Configuration]
 
A \emph{rotor configuration} (or simply configuration) of a rotor graph $G$ is a mapping $\rho$ from $V_0$ to ${\cal A}$ such that $\rho(u) \in {\cal A}^+(u)$ for all $u \in V_0$. We denote by ${\cal C}(G)$ the set of all rotor configurations on the rotor graph $G$.
\end{definition}

What will be called a \emph{particle} in the remaining of this paper is a pebble which will move from one vertex to another; hence the position of the particle is characterized by a single vertex. This movement, called rotor walk, follows specific rules that we detail after.

\begin{definition}[Rotor-particle configuration]
A \emph{rotor-particle configuration}  is a couple $(\rho,u)$ where $\rho$ is a rotor configuration and $u \in V$ denotes the position of a particle.
\end{definition}

\subsubsection*{Rotor Walk}

\begin{definition}
Let us define two mappings on  ${\cal C}(G) \times V_0$ :

\begin{itemize}
\item $\mathbf{turn}$ :
 $$ \mathbf{turn} :  {\cal C}(G) \times V_0 \longrightarrow  {\cal C}(G) $$
 is defined by 
 $$ \mathbf{turn} (\rho, u) = \rho'$$
 where $\rho'$ is equal to $\rho$ except at $v$ where $\rho'(u)=\theta_u(\rho(u))$.

 \item \textbf{\emph{move}} :
 $$ \mathbf{move} :  {\cal C}(G) \times V_0 \longrightarrow V $$
 is defined by 
 $$\mathbf{move} (\rho, u) = h(\rho(u))$$

\end{itemize}

\end{definition}

By composing those mappings, we are now ready to define the \emph{routing of a particle} which is a single step of a rotor walk.

\begin{definition}[Routing of a Particle]
The \textbf{routing} of a particle from a rotor-particle configuration $(\rho,u)$ is a mapping: 
$$ \mathbf{routing} :  {\cal C}(G) \times V_0 \longrightarrow {\cal C}(G) \times V $$
 defined by 
 $$\mathbf{routing} (\rho, u) = (\rho',v)$$ with $\rho'=\mathbf{turn}(\rho,u)$ and $v=\mathbf{move}(\rho,u)$. This can be viewed as the particle first travelling through $\rho(u)$ and then $\rho(u)$ is replaced by $\theta_u(\rho(u))$. This operation is illustrated in \autoref{fig :RotorRouting}.
\end{definition}

\begin{remark}
  Our routing rule (move, then turn) is slightly different than the one defined in \cite{priezzhev1996eulerian} which is mostly used in the literature (turn, then move) but is more convenient to study ARRIVAL. The two
rules are equivalent up to applying $\mathbf{turn}$ mapping on all vertices.
\end{remark}

\begin{figure}
\centering
    \begin{subfigure}[t]{0.47\textwidth}
        \begin{adjustbox}{max totalsize={.8\textwidth}{0.7\textheight},center}
        \begin{tikzpicture}

    \node[shape=circle,draw=black] (A) at (0,0) {$u_0$};
    \node[shape=circle,draw=black] (B)at (3.5,2) {$u_1$};
    \node[shape=circle,draw=black] (C) at (3.5,-2) {$u_2$};
    \node[shape=circle,draw=black, thick] (S1') at (7,2) {$s_1$};
    \node[shape=circle,draw=black] (S1) at (7,2) {$~~~~~~$};
    \node[shape=circle,draw=black,thick] (S2') at (7,-2) {$s_2$};
    \node[shape=circle,draw=black] (S2) at (7,-2) {$~~~~~~$};
    
    \path [->, >=latex](B) edge (S1);
    \path [->, >=latex](C) edge (S2);
    \path [->, >=latex](C) edge (A);
    \path [->, >=latex](A) edge (B);
    \path [->, >=latex](B) edge  (C);
    \path [->, >=latex, very thick,dashed, red,bend right=50](A) edge (C);
    \path [->, >=latex,dashed, bend right=50, very thick, red](B) edge (A);
    \path [->, >=latex,dashed, very thick, bend right=50, red](C) edge  (B);
    \node at (C){
    \includegraphics[width=0.69cm]{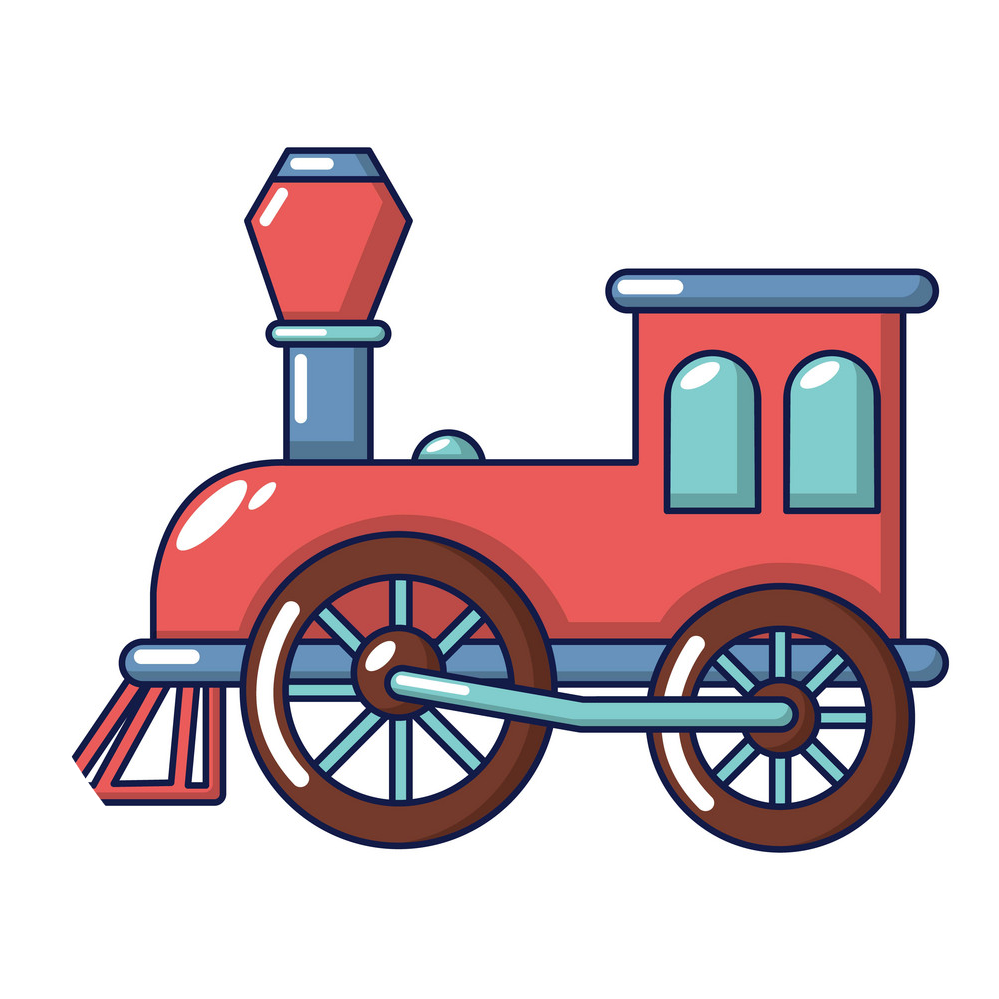} 
    };
  
    \node at (0.5,0.6) {\scriptsize{1}};
    \node at (0.5,-0.6) {\scriptsize{0}};
    \node at (3.7,1.2) {\scriptsize{1}};
    \node at (2.88,2.4) {\scriptsize{0}};
    \node at (4.1,2.3) {\scriptsize{2}};
    \node at (3.9,-1.32) {\scriptsize{0}};
    \node at (2.88,-1.9) {\scriptsize{1}};
    \node at (4.1,-2.3) {\scriptsize{2}};
    
    \draw [ thick,->,>=stealth, red](6.5,0.5) arc (0:330:0.4cm);

    \end{tikzpicture}
        \end{adjustbox}
        \caption{Let $\rho$ be the rotor configuration depicted by the red arcs in dashes.}
    \end{subfigure}
    ~\vline~
    \begin{subfigure}[t]{0.47\textwidth}
        \begin{adjustbox}{max totalsize={.8\textwidth}{0.7\textheight},center}
        \begin{tikzpicture}
    
        \node[shape=circle,draw=black] (A) at (0,0) {$u_0$};
        \node[shape=circle,draw=black] (B) at (3.5,2) {$u_1$};
        \node[shape=circle,draw=black] (C) at (3.5,-2) {$u_2$};
        \node[shape=circle,draw=black, thick] (S1') at (7,2) {$s_1$};
        \node[shape=circle,draw=black] (S1) at (7,2) {$~~~~~~$};
        \node[shape=circle,draw=black,thick] (S2') at (7,-2) {$s_2$};
        \node[shape=circle,draw=black] (S2) at (7,-2) {$~~~~~~$};
        \path [->, >=latex](B) edge (S1);

        \path [->, >=latex](C) edge (S2);
        \path [->, >=latex,very thick, red,dashed](C) edge (A);
        \path [->, >=latex](A) edge (B);
        \path [->, >=latex](B) edge  (C);
        \path [->, >=latex, very thick,dashed, red, bend right=50](A) edge (C);
        \path [->, >=latex, bend right=50, very thick, red,dashed](B) edge (A);
        \path [->, >=latex, very thick, bend right=50](C) edge  (B);
        \node at (B) {
        \includegraphics[width=0.69cm]{trainRond.png}
        };
        \draw [ thick,->,>=stealth, dashed](4.4,0) arc (90:145:2.7cm);
        
        \node at (0.5,0.6) {\scriptsize{1}};
        \node at (0.5,-0.6) {\scriptsize{0}};
        \node at (3.7,1.2) {\scriptsize{1}};
        \node at (2.88,2.4) {\scriptsize{0}};
        \node at (4.1,2.3) {\scriptsize{2}};
        \node at (3.9,-1.32) {\scriptsize{0}};
        \node at (2.88,-1.9) {\scriptsize{1}};
        \node at (4.1,-2.3) {\scriptsize{2}};
    \end{tikzpicture}
        \end{adjustbox}
    \caption{The red rotor configuration in dashes is obtained by processing the operation : $\mathbf{routing} (\rho, u_2)$}
    \end{subfigure}  
     \caption{A rotor-routing where the particle is depicted by a train and starts on $u_2$. The  sink-vertices are $s_1$ and $s_2$. The red arcs in dashes represent the current rotor configuration. The rotor orders on the different vertices are anticlockwise, i.e. they are: $\theta_{u_0}$:($u_0$,$u_2$),($u_0$,$u_1$); $\theta_{u_1}$: ($u_1$,$u_0$),($u_1$,$u_2$),($u_1$,$s_1$); $\theta_{u_2}$: ($u_2$,$u_1$),($u_2$,$u_0$),($u_2$,$s_2$).
        These orders are also depicted by the numbers around each vertex.}
        \label{fig :RotorRouting}
\end{figure}

A rotor-routing is in fact a single step of a rotor walk.

\begin{definition}[Rotor Walk]
  
  A \emph{rotor walk} is a (finite or infinite) sequence of rotor-particle configurations $(\rho_i, u_i)_{i \geq 0}$, which is recursively defined by $(\rho_{i+1},u_{i+1})=\mathbf{routing} (\rho_i, u_i)$ as long as $u_i \in V_0$.
\end{definition}

One can check that the sequence of vertices in the rotor walk starting from the rotor-particle configuration depicted on \autoref{fig :RotorRouting}(b) until a sink is reached is $u_1,u_0,u_2,u_0,u_1,u_2,s_2$.

\subsubsection*{Routing to Sinks}
 
 Now that we have defined our version of rotor-walk, we proceed to the corresponding version of the problem \emph{ARRIVAL}. Indeed, \emph{ARRIVAL} was defined in \cite{dohrau2017arrival} using the traditional rotor routing operation where we turn first and then move.

 \begin{definition}[Maximal Rotor Walk]
    A maximal rotor walk is a rotor walk such that in the case where it is finite, the last vertex must be a sink vertex $s \in S_0 = V \setminus V_0$.
\end{definition}

\begin{definition}[Stopping Rotor Graph]
If, for any vertex $u \in V$, there exists a directed path from $u$ to a sink $s \in S_0$, the graph is said to be \emph{stopping}.
\end{definition}

The next lemma is a classical result on rotor walks (cf Lemma 16 in \cite{giacaglia2011local}).

\begin{lemma}[Finite number of steps]
\label{lem:finiteRotorWalk}
If $G$ is stopping then any rotor walk in $G$ is finite. 
\end{lemma}

\begin{proof}
Given a rotor-particle configuration $(\rho,u)$ with $u \in V_0$, let $R$ be a rotor walk starting from $(\rho,u)$. In this walk, a sink vertex of $S_0$ can only appear as the last vertex, hence it can be visited only once. Now, let us consider $S_1=$ \inneighbours{S_0}. Each time a vertex $t \in S_1$ is visited, the particle moves to the next neighbour. Hence, it can be visited at most $\lvert {\cal A}^+(t) \rvert$ times before a sink is reached and the walk ends consequently. Now, let us consider the set $S_2=$ \inneighbours{S_1}, all vertices $u$ within $S_2$ can only be visited a finite number of times, otherwise the particle will visit a vertex from $S_1$ more than $\lvert {\cal A}^+(t) \rvert$ times. We proceed inductively to show that all vertices can be visited only a finite number of times in $R$ since all vertices of $G$ have a directed path to a sink.
\end{proof}

This proof provides a bound on the number of times an arc has been visited in a finite rotor walk (\autoref{returnBound}).

The main objective of our work is to study the sink that will be reached by a maximal rotor walk from a rotor-particle configuration, if the rotor walk is finite.

\begin{definition}[Exit Sink]
Let $u \in V$, let $\rho$ be a configuration, if the maximal rotor walk starting from $(\rho,u)$ is finite in $G$, then the sink reached by such a rotor walk is denoted by \exitSink{G}{$\rho$}{u} and called \emph{exit sink} of $u$ for the rotor configuration $\rho$ in $G$.
\end{definition}

\begin{definition}[Exit Pattern]
For a rotor configuration $\rho$ on a stopping rotor graph $G$, the \emph{exit pattern} is the mapping that associates to each vertex $u \in V$, its exit sink \exitSink{G}{$\rho$}{u}.
\end{definition}
 
\subsection{ARRIVAL and Complexity Issues}
 
With our notations, ARRIVAL (see~\cite{dohrau2017arrival}) can be expressed as the following decision problem:
\setlength{\fboxsep}{2mm}% définir l'écart
\setlength{\fboxrule}{0.3mm} % définir l'épaisseur du trait
\begin{center}
\fbox{ %fbox est utilisé pour voir les bords de la minipage
\begin{minipage}[c]{7cm}

\begin{center} In a stopping rotor graph $G$,

    given $(\rho,u)$ and $s \in S_0$
    
    does $\text{\exitSink{G}{$\rho$}{u}}=s$ ?
\end{center}

\end{minipage}
}
\end{center}

 This problem belongs to NP $\cap$ co-NP for simple graphs as shown in \cite{dohrau2017arrival}, but there is still no polynomial algorithm known to solve it. The case of eulerian simple graphs can be solved in time  O($|V+{\cal A}|^3$), since a finite maximal rotor walk from $(\rho,v)$ ends in at most O($|V+{\cal A}|^3$) routings (see~\cite{yanovski2003distributed}).

 In the case of multigraphs,  ARRIVAL still belongs to NP $\cap$ co-NP, since the polynomial certificate used for simple graphs in \cite{dohrau2017arrival} remains valid. Our goal in this paper is to solve the problem in polynomial time for a particular class of multigraphs. Despite that, even for a path multigraph, the routing can be exponential, as shown on \autoref{expchain}.

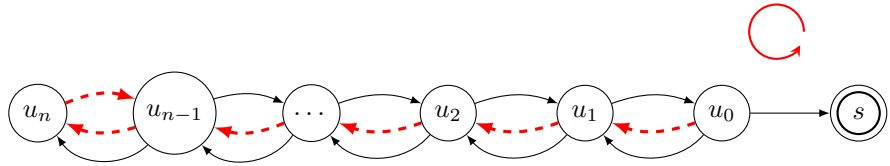
\begin{figure}[t!]
       \centering
        \begin{tikzpicture}[scale=0.9]

    \node[shape=circle,draw=black] (A) at (-4,0) {$u_n$};
    \node[shape=circle,draw=black] (A') at (-2,0) {$u_{n-1}$};
    \node[shape=circle,draw=black] (B') at (0,0) {$\dots$};
    \node[shape=circle,draw=black] (B) at (2,0) {$u_2$};
    \node[shape=circle,draw=black] (C) at (4,0) {$u_1$};
    \node[shape=circle,draw=black] (D) at (6,0) {$u_0$};
    \node[shape=circle,draw=black,thick] (E2) at (8,0) {$s$};
    \node[shape=circle,draw=black] (E) at (8,0) {$~~~~$};

    \path[->, >=latex,dashed, very thick,red, bend left=20](A) edge (A');
   \path [->, >=latex,dashed, very thick, red, bend left=20](A') edge (A);
   \path [->, >=latex, bend left=50](A') edge (A);
   \path[->, >=latex, bend left=20](A') edge (B');
   \path [->, >=latex,dashed, very thick, red, bend left=20](B') edge (A');
   \path [->, >=latex, bend left=50](B') edge (A');
   \path[->, >=latex, bend left=20](B') edge (B);
   \path [->, >=latex,dashed, very thick, red, bend left=20](B) edge (B');
   \path [->, >=latex, bend left=50](B) edge (B');
   \path [->, >=latex,dashed, very thick, red, bend left=20](C) edge (B);
    \path [->, >=latex, bend left=50](C) edge (B);
   \path [->, >=latex,  bend left=20](B) edge (C);
   \path [->, >=latex,dashed, very thick, red, bend left=20](D) edge (C);
   \path [->, >=latex, bend left=50](D) edge (C);
   \path [->, >=latex,  bend left=20](C) edge (D);
\path [->, >=latex](D) edge (E);

\draw [ thick,->,>=stealth, red](7.2,1.2) arc (0:330:0.4cm);

\end{tikzpicture}
        \caption{Family of path-like multigraphs where maximal routing can take an exponential number of steps in the number of vertices, here equal to $n+2$. 
        The interior vertices ($u_0$ to $u_{n-1}$) have two arcs going left
        and one going right.
        Routing a particle from $u_0$ to  sink
        $s$ with the initial configuration $\rho$ drawn with red arcs in dashes takes a non-polynomial time considering the anticlockwise rotor ordering on each vertex, depicted by the curved arc in red. }
        \label{expchain}
    \end{figure}

\smallskip

 In the example drawn in \autoref{expchain}, let us consider the number of times the particle will travel from $u_i$ to $u_{i+1}$. Let $u_0$ be the starting vertex: to reach the sink $s$, the particle needs to visit $u_0$ exactly three times and so it will travel from $u_0$ to $u_1$ exactly two times. Next, for $u_1$, each time the particle comes from $u_0$, it will travel two times from $u_1$ to $u_2$ before visiting $u_0$ again. So the number of times the particle travels from $u_1$ to $u_2$ is 4. One can check that the number of times a particle starting from $u_0$ will travel from $u_i$ to $u_{i+1}$ before reaching $s$ is $2^{i+1}$.

 \subsection{Cycle Pushing}
 In order to speed-up the rotor walk process, a simple tool to avoid computing every step of the walk is the use of \emph{cycle pushing}. We keep the terminology of \emph{cycle pushing} used in \cite{giacaglia2011local} -- even if what is called a cycle in this terminology is usually called a directed cycle or circuit in graph theory.
 
 \begin{definition}[Graph Induced by $\rho$]
 \label{inducedGraph} If $G$ is a rotor multigraph and $\rho \in {\cal C}(G)$,
 we denote by $G(\rho)=(V_0,S_0,\rho(V_0),h,t)$ the \emph{graph induced} by $\rho$ on $G$. This means that $G(\rho)$ and $G$ have the same set of vertices, but $G(\rho)$ only contains arcs of $\rho(V_0)$. In the following, we will say that the (directed) cycle $C$ is \emph{in} $\rho$ if it belongs to $G(\rho)$.
 \end{definition}
 
 \begin{remark}
  All non-sink vertices of $G(\rho)$ are of outdegree exactly one, hence a vertex 
  belongs to at most one cycle which is in $\rho$.
 \end{remark}

 \begin{definition}[Cycle $Pushing$]
Let $C=(u_1, u_2, u_3, ..., u_k)$ be a directed cycle of $G(\rho)$. We call \emph{cycle push} the operation on $\rho$ that leads to a configuration $\rho'$ such that: 
\begin{itemize}
    \item for all $u_i \in C$,  $\rho'(u_i)=\theta_{u_i}(\rho(u_i))$;
    \item $\forall u \in V_0 \setminus C, \rho'(u)=\rho(u)$
\end{itemize}
i.e. we process a \textbf{turn} on all $u_i \in C$.
\end{definition}

Note that a cycle push on a cycle $C$ can also be computed by putting a particle on a vertex $u$ of $C$ and routing the particle until it comes back to $u$ for the first time. Hence, cycle pushing is in a sense a shortcut on the rotor walk process, so by a manner similar to \autoref{lem:finiteRotorWalk} it follows that:

\begin{lemma}[Finite Number of Cycle Pushes]
\label{finitePush}
Given a stopping rotor graph, any sequence of cycle pushes is finite.
\end{lemma}

The previous result is well known in rotor walk studies (\cite{Holroyd2008}). It implies that, by processing a long enough sequence
of successive cycle pushes, the resulting configuration contains no directed cycles. Such a
sequence of cycle pushes is called \emph{maximal}.

The two following results can be found in \cite{giacaglia2011local}.

\begin{lemma}[Exit Pattern conservation for Cycle Push]
\label{sinkconserv} If $G$ is a stopping rotor graph,
for any rotor configuration $\rho$ and configuration $\rho'$ obtained from $\rho$ by a cycle push , the exit pattern for $\rho$ and $\rho'$ is the same.
\end{lemma}

\begin{lemma}[Commutativity of Cycle Push]
In a stopping rotor graph, any sequence of cycle pushes starting from a given rotor configuration $\rho$ leads to the same maximal, acyclic configuration $\rho'$ (i.e. $G(\rho')$ is an acyclic graph).
\end{lemma}

\begin{definition}[Destination Forest]
We call the configuration obtained by a maximal cycle push sequence on $\rho$ the \emph{Destination Forest} of $\rho$, denoted by $D(\rho)$.
\end{definition}

The destination forest has a simple interpretation in terms of rotor walks: start a rotor walk by putting a particle on any vertex of a stopping graph $G$; consider a vertex $u \in V_0$; if the particle ever reaches $u$, it will leave $u$ by arc $D(\rho)(u)$ on the last time it enters $u$.

In an acyclic configuration like $D(\rho)$, finding the exit pattern is
very simple, precisely:

\begin{lemma}[Path to a sink]
If there is a directed path between $u \in V$ and $s \in S_0$ in $G(\rho)$ then $\text{\exitSink{G}{$\rho$}{u}}=s$. It follows that from $D(\rho)$ one can compute the exit pattern of $\rho$ in time complexity $O(|\cal{A}|)$.
\end{lemma}

This gives us a new approach, since computing the exit pattern of a configuration $\rho$ can be done by computing its Destination Forest $D(\rho)$. Note that by doing this we are solving a problem harder than ARRIVAL because we compute the exit sink of all vertices simultaneously.

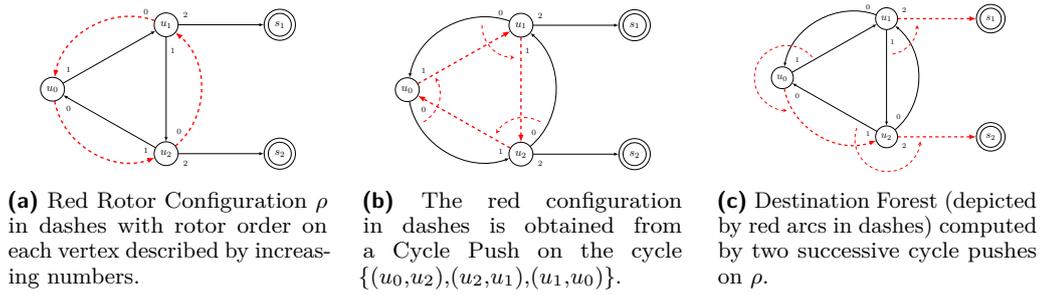
\begin{figure}
    \centering
    \begin{subfigure}[t]{0.3\textwidth}
    \centering
        \begin{adjustbox}{max totalsize={.8\textwidth}{0.7\textheight},center}
        \begin{tikzpicture}
    \node[shape=circle,draw=black] (A) at (0,0) {$u_0$};
    \node[shape=circle,draw=black] (B)at (3.5,2) {$u_1$};
    \node[shape=circle,draw=black] (C) at (3.5,-2) {$u_2$};
    \node[shape=circle,draw=black] (S1') at (7,2) {$s_1$};
    \node[shape=circle,draw=black] (S2') at (7,-2) {$s_2$};
    \node[shape=circle,draw=black] (S1) at (7,2) {$~~~~~~$};
    \node[shape=circle,draw=black] (S2) at (7,-2) {$~~~~~~$};
    
    \path [->, >=latex](B) edge (S1);
    \path [->, >=latex](C) edge (S2);
    \path [->, >=latex](C) edge (A);
    \path [->, >=latex](A) edge (B);
    \path [->, >=latex](B) edge  (C);
    \path [->, >=latex, very thick,dashed, red,bend right=50](A) edge (C);
    \path [->, >=latex,dashed, bend right=50, very thick, red](B) edge (A);
    \path [->, >=latex,dashed, very thick, bend right=50, red](C) edge  (B);
    \node at (C){
    };
  
    \node at (0.5,0.6) {\scriptsize{1}};
    \node at (0.5,-0.6) {\scriptsize{0}};
    \node at (3.7,1.2) {\scriptsize{1}};
    \node at (2.88,2.4) {\scriptsize{0}};
    \node at (4.1,2.3) {\scriptsize{2}};
    \node at (3.9,-1.32) {\scriptsize{0}};
    \node at (2.88,-1.9) {\scriptsize{1}};
    \node at (4.1,-2.3) {\scriptsize{2}};
 
 \end{tikzpicture}
        \end{adjustbox}
        \caption{Red Rotor Configuration $\rho$ in dashes with rotor order on each vertex described by increasing numbers.}
    \end{subfigure}
    ~~
      \begin{subfigure}[t]{0.3\textwidth}
        \begin{adjustbox}{max totalsize={.8\textwidth}{0.7\textheight},center}
         \begin{tikzpicture}
    \node[shape=circle,draw=black] (A) at (0,0) {$u_0$};
    \node[shape=circle,draw=black] (B)at (3.5,2) {$u_1$};
    \node[shape=circle,draw=black] (C) at (3.5,-2) {$u_2$};
    \node[shape=circle,draw=black] (S1') at (7,2) {$s_1$};
    \node[shape=circle,draw=black] (S2') at (7,-2) {$s_2$};
    \node[shape=circle,draw=black] (S1) at (7,2) {$~~~~~~$};
    \node[shape=circle,draw=black] (S2) at (7,-2) {$~~~~~~$};
    \node at (0.5,0.6) {\scriptsize{1}};
    \node at (0.5,-0.6) {\scriptsize{0}};
    \node at (3.7,1.2) {\scriptsize{1}};
    \node at (2.88,2.4) {\scriptsize{0}};
    \node at (4.1,2.3) {\scriptsize{2}};
    \node at (3.9,-1.32) {\scriptsize{0}};
    \node at (2.88,-1.9) {\scriptsize{1}};
    \node at (4.1,-2.3) {\scriptsize{2}};
    
    \path [->, >=latex](B) edge (S1);
    \path [->, >=latex](C) edge (S2);
    \path [->, >=latex, dashed, red, very thick](C) edge (A);
    \path [->, >=latex, red,dashed, very thick](A) edge (B);
    \path [->, >=latex, red,dashed, very thick](B) edge  (C);
    \path [->, >=latex, bend right=50 ](A) edge (C);
    \path [->, >=latex, bend right=50](B) edge (A);
    \path [->, >=latex, bend right=50](C) edge  (B);
    \draw [ thick,->,>=stealth, red, dashed](2.3,2) arc (180:270: 1cm);
    \draw [ thick,->,>=stealth, red, dashed](0.5,-1) arc (-60:30: 1cm);
    \draw [ thick,->,>=stealth, red, dashed](4.1,-1) arc (60:150: 1cm);

\end{tikzpicture}
        \end{adjustbox}
        \caption{The red configuration in dashes is obtained from a Cycle Push on the cycle \{($u_0$,$u_2$),($u_2$,$u_1$),($u_1$,$u_0$)\}.}
    \end{subfigure}
    ~~
    \begin{subfigure}[t]{0.3\textwidth}
        \begin{adjustbox}{max totalsize={.8\textwidth}{0.7\textheight},center}
        \begin{tikzpicture}
     \node[shape=circle,draw=black] (A) at (0,0) {$u_0$};
    \node[shape=circle,draw=black] (B)at (3.5,2) {$u_1$};
    \node[shape=circle,draw=black] (C) at (3.5,-2) {$u_2$};
    \node[shape=circle,draw=black] (S1') at (7,2) {$s_1$};
    \node[shape=circle,draw=black] (S2') at (7,-2) {$s_2$};
    \node[shape=circle,draw=black] (S1) at (7,2) {$~~~~~~$};
    \node[shape=circle,draw=black] (S2) at (7,-2) {$~~~~~~$};
    \node at (0.5,0.6) {\scriptsize{1}};
    \node at (0.5,-0.6) {\scriptsize{0}};
    \node at (3.7,1.2) {\scriptsize{1}};
    \node at (2.88,2.4) {\scriptsize{0}};
    \node at (4.1,2.3) {\scriptsize{2}};
    \node at (3.9,-1.32) {\scriptsize{0}};
   \node at (2.88,-1.9) {\scriptsize{1}};
    \node at (4.1,-2.3) {\scriptsize{2}};

    \path [->, >=latex, dashed, red, very thick](B) edge (S1);
    \path [->, >=latex, dashed, red, very thick](C) edge (S2);
    \path [->, >=latex](C) edge (A);
    \path [->, >=latex](A) edge (B);
    \path [->, >=latex](B) edge  (C);
    \path [->, >=latex, bend right=50, dashed, red, very thick ](A) edge (C);
    \path [->, >=latex, bend right=50](B) edge (A);
    \path [->, >=latex, bend right=50](C) edge  (B);
    
    \draw [ thick,->,>=stealth, red, dashed](3.7,0.8) arc (280:360: 1cm);
    \draw [ thick,->,>=stealth, red, dashed](1,0.7) arc (40:270: 1.1cm);
    \draw [ thick,->,>=stealth, red, dashed](2.5,-1.7) arc (160:356: 1.1cm);

\end{tikzpicture}
        \end{adjustbox}
         \caption{Destination Forest (depicted by red arcs in dashes) computed by two successive cycle pushes on $\rho$.}
    \end{subfigure}
\caption{Computation of the Destination Forest by successive cycle pushing.}

\end{figure}

\begin{remark}
Pushing a cycle of length $k$ is a "shortcut" in a rotor walk as it allows to do only one operation to simulate $k$ steps of the rotor walk. However, the strategy consisting in pushing cycles until the Destination Forest is reached (which is always the case if the graph is stopping) can still take an exponential time. Indeed, let us consider \autoref{expchain} again. As each cycle in $G(\rho)$ is at most of size two, we are at most dividing by two the number of steps used in the maximal rotor walk.
\end{remark}

So, processing successive cycle pushes is not a suitable strategy to solve ARRIVAL efficiently in general. However, the notion of equivalence classes that we state below, which is useful for a deeper understanding of rotor walks, is based on cycle pushing.

\begin{definition}[Equivalence Class \cite{giacaglia2011local}]
\label{EquivClass}
We define an equivalence relation between configurations by considering that two configurations $\rho_1$ and $\rho_2$ are in the same class if they have the same Destination Forest.
\end{definition}

The proof of the following theorem can be found in \cite{giacaglia2011local}, Corollary 14.
It shows that equivalence of configurations is preserved by the routing of particle to a sink.

\begin{theorem}
\label{thm:permut_classes}
 Let $\rho_1, \rho_2$ be two equivalent configurations. Let $\rho'_1$
and $\rho'_2$ be the configurations obtained respectively from 
 $\rho_1$ and $\rho_2$ by routing a single particle from the same 
 vertex $u$ to a sink $s$. Then $\rho'_1$ and $\rho'_2$ are equivalent.
 \end{theorem}

\section{The Simple Path Graph} \label{sec:path}

In this section we give some results on a particular graph in order to give  intuition behind our following work on trees.
The simple path graph on $n+2$ vertices is defined by
$V = \{s_0, u_1, u_2, \dots, u_n, s_1  \} $ and each $u_i$ for $i \in \{1,..,n\}$ has an arc going to the vertex on the left and another one going to the vertex on the right.
Define then $ V_0 = \{u_1,u_2,...,u_n\}$ together with $S_0=\lbrace s_0,s_1\rbrace$.

See \autoref{fig:chainDFComputation} for the simple path graph with $n=4$.
\vspace{0.3cm}

 In this section we will not only show that ARRIVAL can be solved in linear time on simple path graphs but that we can solve a harder problem by calculating the exit pattern for $\rho$ in linear time, i.e solving ARRIVAL simultaneously for all vertices.
 
 However, note that the maximum number of steps of a rotor walk (or a cycle push sequence) in a simple path graph is not exponential but quadratic in $n$. Indeed, one can check that the starting configuration that maximizes this number of steps is the configuration where all arcs are directed towards the central vertex $u_{\lceil n/2 \rceil}$ (which is also the starting vertex in the case of the rotor walk).

\subsection{Routing One Particle on a Path Graph}

It is an easy observation that the only elementary cycles in this graph are of length 2, and that pushing such a cycle in a given configuration does not change
the global numbers of arcs respectively directed towards $s_0$ and $s_1$ in the configuration. Hence the following definition:

\begin{definition}
We say that a vertex $u_i \in V_0$ is \emph{directed towards} $s_0$ if $h(\rho(u_i)) \in \{u_{i-1},s_0\}$. Otherwise, $u_i$ is directed towards $s_1$ (see \autoref{chainExea}).
We denote by $n_1(\rho)$ the number of vertices of $V_0$ that are directed towards $s_1$.
\end{definition}

From the observation above will follow that we can compute the exit pattern of $\rho$ without having to process the sequence of cycle pushes, relying only on the computation of $n_1(\rho)$.

\begin{lemma}[Exit Sink Characterization]
\label{ExitCharact}
Two configurations are equivalent if and only if
they have the same number $n_1$. In particular,
for any configuration $\rho$, if $1 \leq i \leq n-n_1(\rho)$, then $\text{\exitSink{G}{$\rho$}{u_i}}=s_0$, otherwise $\text{\exitSink{G}{$\rho$}{u_i}}=s_1$.
\end{lemma}

\begin{proof}
A maximal cycle push sequence leads to an acyclic configuration, and we know that $n_1$ is preserved by cycle push. Since there is a unique forest with $n_1$ arcs directed towards $s_1$, namely the configuration where exactly $u_{n-(n_1-1)}, u_{n-(n_1-2)}, \dots, u_{n}$ are directed towards $s_1$, this is the Destination Forest (see \autoref{chainExeb}).
\end{proof}

The previous lemma enables us to compute in linear time the exit pattern of any configuration in a simple path graph -- which easily solves ARRIVAL simultaneously for all starting vertices. 

 \begin{figure}
 
    \begin{center}
    
    \begin{subfigure}[c]{\textwidth}
    \centering
    \begin{tikzpicture}
    \node[shape=circle,draw=black, thick] (A') at (0,0) {$s_0$};
    \node[shape=circle,draw=black] (B) at (2,0) {$u_1$};
    \node[shape=circle,draw=black] (C) at (4,0) {$u_2$};
    \node[shape=circle,draw=black] (D) at (6,0) {$u_3$};
    \node[shape=circle,draw=black] (E) at (8,0) {$u_4$};
    \node[shape=circle,draw=black, thick] (F') at (10,0) {$s_1$};
    \node[shape=circle,draw=black] (A) at (0,0) {$~~~~~~$};
     \node[shape=circle,draw=black] (F) at (10,0) {$~~~~~~$};

   \path [->, >=latex, bend left=30,dashed, red, very thick](B) edge (C);
   \path [->, >=latex, bend left=30](C) edge (B);
   \path [->, >=latex, bend left=30,very thick,dashed, red](E) edge (D);
   \path [->, >=latex, bend left=30](D) edge (E);
   \path [->, >=latex, bend left=30,dashed, red, very thick](C) edge (D);
   \path [->, >=latex, bend left=30,dashed, red, very thick](D) edge (C);
   \path[->,>=latex](B) edge (A);
   \path[->,>=latex](E) edge (F);
   
  \end{tikzpicture}
  \caption{$u_1$ and $u_2$ are directed towards $s_1$,  and $u_3$ and $u_4$ are directed towards $s_0$. Let $\rho$ be the initial configuration depicted here. }
  \label{chainExea}
    \end{subfigure}
    
    \begin{subfigure}[c]{\textwidth}
 \centering   
\begin{tikzpicture}

 \node[shape=circle,draw=black, thick] (A') at (0,0) {$s_0$};
    \node[shape=circle,draw=black] (B) at (2,0) {$u_1$};
    \node[shape=circle,draw=black] (C) at (4,0) {$u_2$};
    \node[shape=circle,draw=black] (D) at (6,0) {$u_3$};
    \node[shape=circle,draw=black] (E) at (8,0) {$u_4$};
    \node[shape=circle,draw=black, thick] (F') at (10,0) {$s_1$};
    \node[shape=circle,draw=black] (A) at (0,0) {$~~~~~~$};
     \node[shape=circle,draw=black] (F) at (10,0) {$~~~~~~$};

   \path [->, >=latex, bend left=30, red, dashed,very thick](B) edge (C);
   \path [->, >=latex, bend left=30, red,dashed, very thick](C) edge (B);
   \path [->, >=latex, bend left=30, red, dashed, very thick](E) edge (D);
   \path [->, >=latex, bend left=30, red,dashed, very thick](D) edge (E);
  \path [->, >=latex, bend left=30](C) edge (D);
   \path [->, >=latex, bend left=30](D) edge (C);
   \path[->,>=latex](B) edge (A);
   \path[->,>=latex](E) edge (F);
   
 \end{tikzpicture}
 \caption{Configuration obtained from $\rho$ after a single cycle push on the cycle $\{(u_2,u_3);(u_3,u_2)\}$. This configuration and $\rho$ are equivalent since one is obtained from the other by cycle pushing so they have the same Destination Forest. And as stated in \autoref{ExitCharact} they have the same value $n_1=2$.}
 \end{subfigure}
    
    \begin{subfigure}[c]{\textwidth}
 \centering   
\begin{tikzpicture}

 \node[shape=circle,draw=black, thick] (A') at (0,0) {$s_0$};
    \node[shape=circle,draw=black] (B) at (2,0) {$u_1$};
    \node[shape=circle,draw=black] (C) at (4,0) {$u_2$};
    \node[shape=circle,draw=black] (D) at (6,0) {$u_3$};
    \node[shape=circle,draw=black] (E) at (8,0) {$u_4$};
    \node[shape=circle,draw=black, thick] (F') at (10,0) {$s_1$};
    \node[shape=circle,draw=black] (A) at (0,0) {$~~~~~~$};
     \node[shape=circle,draw=black] (F) at (10,0) {$~~~~~~$};

   \path [->, >=latex, bend left=30](B) edge (C);
   \path [->, >=latex, bend left=30, red,dashed, very thick](C) edge (B);
   \path [->, >=latex, bend left=30](E) edge (D);
   \path [->, >=latex, bend left=30, red,dashed, very thick](D) edge (E);
  \path [->, >=latex, bend left=30](C) edge (D);
   \path [->, >=latex, bend left=30](D) edge (C);
   \path[->,>=latex, red,dashed, very thick](B) edge (A);
   \path[->,>=latex, red,dashed, very thick](E) edge (F);
   
 \end{tikzpicture}
 \caption{Naturally, the Destination Forest (i.e. rotor configuration after maximal cycle push sequence) of $\rho$ is equivalent to both previous configurations and the value $n_1$ remains the same.}
 \label{chainExeb}
 \end{subfigure}
 
 \caption{Illustration of \autoref{ExitCharact} and its proof. The red arcs in dashes represent the current rotor configuration.}
 \label{fig:chainDFComputation}
 \end{center}

 \end{figure}
 
\subsection{Routing Several Particles}
\label{sec:multiparticles}

Here we present how the previous results can be used to solve ARRIVAL on a simple path graph if we want to route multiple particles. We only study the case of routing multiple particles in this case, the general case remaining an open problem even for tree-like multigraphs.

It is well known that when we route any number of particles in a rotor graph, the final configuration does not depend on the order in which the particles move (which can alternate between particles) as long as we route the particles to the sinks (see \cite{Holroyd2008}).

\begin{definition}[Equivalence Class for simple path graphs]
For $0 \leq k \leq n$, let $C_k$ be the class of configurations $\rho$
with $n_1(\rho) = k$ i.e. the set of configurations with exactly $k$ vertices directed towards $s_1$.
\end{definition}

Note that the equivalence classes defined just above are exactly the equivalence classes of \autoref{EquivClass}. 
We now explain the action of the movement of a particle on these classes.

\begin{theorem}[Group Action on simple path]
\label{thm:GroupActionPath}
Consider a configuration $\rho \in C_k$ with $0 \leq k \leq n$
and a vertex $u_i$, $1 \leq i \leq n$. Processing a maximal rotor walk from $(\rho,u_i)$ leads to a configuration $\rho' \in C_j$ where $j=(k+i) \mod n+1$.
\end{theorem}

\begin{proof}
By \autoref{thm:permut_classes} we can as well suppose that $\rho$ is acyclic, hence
the $k$ vertices  directed towards $s_1$ are exactly $u_{n-k+1}, u_{n-k+2}, \dots, u_n$.
\begin{itemize}
    \item If $i \leq n-k$, then during the rotor walk, exactly
    $i$ vertices change directions, hence $j=k+i$.
    \item If $i > n - k$, then $n-i+1$ vertices change directions, hence $j = k - (n-i+1) = k+i - (n+1)$.
\end{itemize}
\end{proof}

When combining the action of multiple particles on classes, we obtain the following theorem.

\begin{theorem}
Let $\rho \in C_k$ be a rotor configuration. Consider the process where $m$ particles are routed to the sinks from initial positions $i_1, i_2, \dots, i_m$ in any (alternating or not) order. Then:
\begin{itemize}
    \item The final configuration belongs to the class $C_j$ with
\[ j = k + i_1 + i_2 + \dots + i_m \mod n+1,  \]
    \item exactly
    $p_1 = \left\lfloor \dfrac{k +(\sum_{t=1}^{m} i_t)}{n+1}\right\rfloor$
    of these particles reach sink $s_1$, whereas $m - p_1$
    reach $s_0$.
    \end{itemize}
\end{theorem}

\begin{proof}
By commutativity of the process, we can suppose that we fully route every particle to a sink before proceeding to the next one.
Consider $\{0,1,\dots,n\}$ as a cycle of positions corresponding to the equivalence classes. We start from position $k$ on the cycle, and every time a particle is routed from its initial position $i_t$ to a sink,
the final configuration belongs to $C_j$ where $j$ is obtained by moving from $k$ on the cycle $i_t$ times in cyclic ordering. It is contained in the proof of \autoref{thm:GroupActionPath} that the particle ends in $s_1$ if and only if we move from position $n$ to $0$ during this process. Hence, we just have to count the number of times this will happen, which is simply the integer quotient of the final position $k+i_1+\dots+i_m$ by the number of positions $n+1$.
\end{proof}

Note that the previous theorem enables us to solve in linear time a variant of ARRIVAL for multiple particles where we are interested in the number of particles ending in each sink (even with an exponential number of particles, if arithmetic computations are made in constant time).

\subsection{The Return Flow with the Path Graph}

The previous technique to compute the Destination Forest is based on an invariance property of cycle pushes which is specific to the simple path graph. It does not generalize directly to other graphs, which is why we now give an equivalent formulation of the previous results. Instead of counting the number of vertices that are directed toward sinks, we now consider the number of vertices that are directed towards a given vertex $u$ on each side of the path. To do this, we denote by $]u,u'[$ the set of vertices that lie between $u$ and $u'$ in the order of vertices in the path, excluding $u$ and $u'$.

\begin{definition}
\label{def:r0etr1}
Let $r_1(u)$ and $r_0(u)$ be the number of vertices that are respectively directed towards $u$ in $]u,s_1[$ and $]s_0,u[$.
\end{definition}

We now restate \autoref{ExitCharact} in terms of $r_1$ ans $r_0$.

\begin{lemma}
For all non-sink vertices $u_i$, \exitSink{G}{$\rho$}{u_i}$=s_1$ if
and only if
\begin{itemize}
    \item either $r_1(u_i) < r_0(u_i)$;
    \item or $r_1(u_i) = r_0(u_i)$ and $u_i$ is directed towards $s_1$.
    \end{itemize}
\end{lemma}

\begin{proof}
We can decompose $n_1$ as follow:

$$n_1 = \underbrace{r_0(u_i)}_{\text{vertices of } ]s_0,u_i[} +
\underbrace{I_{1}(u_i)}_{\text{Orientation of } u_i} +
\underbrace{n-i-r_1(u_i)}_{\text{vertices of } ]u_i,s_1[}  $$
where $I_{1}(u_i)$ is 1 if $u_i$ is directed towards $s_1$
and $0$ otherwise.

Now, from \autoref{ExitCharact} we know that $\text{\exitSink{G}{$\rho$}{u_i}}=s_1 \Leftrightarrow i \geq n-n_1+1$, which translates to
\[ r_0(u_i) + I_1(u_i) \geq r_1(u_i) + 1,\]
hence the result.
\end{proof}

\begin{remark}
Consider the following process: on the simple path graph, we remove the arc $(u_i,u_{i-1})$ so that the particle cannot go from $u_i$ to $[s_0,u_i[$ anymore (in the case where $\rho(u_i)=(u_i,u_{i-1})$, we change it to $\rho(u_i)=(u_i,u_{i+1})$). 
Now, we put a particle on $u_{i}$ and proceed to a maximal rotor walk. The number of times that the particle travels through arc $(u_i,u_{i+1})$ is exactly $r_1(u_i)+1$. This quantity is called the \emph{return flow}, which will be defined in the general case in \autoref{DefRN}.
\end{remark}

\section{Tree-Like Multigraphs: Return Flow Definition} \label{sec:return}

\subsection{Tree-Like Multigraphs}

To a directed multigraph $G=(V,{\cal A},h,t)$ we associate: 
\begin{itemize}
    \item A simple directed graph $\hat{G}=(V, \arcs)$  such that, for $u,v \in V$ there is an arc from $u$ to $v$ in $\hat{G}$ if there is at least one arc $a \in {\cal A}$ with $t(a)=u$ and $h(a)=v$. Please note that even if there are multiple arcs $a$ that satisfy this property, there is only one arc with tail $u$ and head $v$ in $\arcs$. As it is unique, an arc from $u$ to $v$ in $\arcs$ will simply be denoted by $(u,v)$.
    \item A simple undirected graph $\overline{G}=(V, E)$ such that, for $u,v \in V$ there is an edge between $u$ and $v$ in $\overline{G}$ if and only if there is at least one arc $a \in {\cal A}$ such that $t(a)=u$ and $h(a)=v$ or $h(a)=u$ and $t(a)=v$.  
\end{itemize} 

\begin{definition}[Tree-Like Multigraph]
We say that a multigraph $G$ is \emph{tree-like} if 
$\overline{G}$ is a tree.
\end{definition}

In this case, we define the leaves of $G$ as the leaves of $\overline{G}$.

\begin{definition}[Tree-Like Rotor Multigraph]
A rotor multigraph $G=(V_0, S_0,{\cal A},h,t, \theta)$ is tree-like if $(V,{\cal A},h,t)$ is tree-like, and its set of leaves contains $S_0$.
\end{definition}

To avoid some complexity in the notation and proofs, we will only study stopping tree-like rotor multigraphs. We first show that the general case can be handled by reducing a non stopping instance to a stopping one.

\begin{definition}[Sink Component]
A sink component is a strongly connected component in $G$ that does not contain a sink vertex and such that there is no arc leaving the component. 
\end{definition}

Note that all sink components can be computed in linear time.

\begin{lemma}
\label{lem:removeSC}
 Consider a configuration $\rho$ on a (not necessarily stopping) tree-like rotor multigraph $G=(V_0, S_0,{\cal A},h,t, \theta)$. Consider a configuration $\rho'$ on the \emph{stopping} tree-like rotor multigraph $G'=(V_0', S_0',{\cal A}',h,t, \theta)$ 
 where $G'$ is obtained from $G$ by replacing each sink component by a unique sink, and where $\rho'(u)=\rho(u)$ for each $u \in V_0'$.
  For any $u \in V$, finding the exit sink of $u$ (if any) or the sink component reached by $u$ in $G$ for the configuration $\rho$ can be directly determined by solving ARRIVAL for the configuration $\rho'$ in $G'$.
\end{lemma}

\begin{proof}
Let $\rho$ be a configuration in ${\cal C}(G)$ and $u \in V$ be a vertex.
\begin{itemize}
    \item If the particle enters a sink component $C$ while processing a rotor walk from $(\rho,u)$ on $G$, then $u$ has no exit sink for $\rho$. The same rotor walk in the graph $G'$ ends in the sink that replaces $C$.
    \item If the particle reaches a sink $s$ of $G$ while processing a maximal rotor walk from $(\rho,u)$ on $G$, then
    it does not enter a sink component of $G$ hence the walk in $G'$ is the same as in $G$.
\end{itemize}
\end{proof}

Thanks to \autoref{lem:removeSC}, we can work on graphs without sink components while keeping the generality of our results. Note that, after replacing sink components by sinks, the multigraph $G'$ may no longer be a tree-like multigraph but a forest-like multigraph. However we can split the study of ARRIVAL in each tree-like component of this forest since a particle cannot travel between those trees in a rotor walk.

\subsection{Return Flows}

Let us consider the simple example depicted on \autoref{IntuitionRN} to motivate the introduction of $(u,v)$-subtrees and return flows, which is our main tool.

\begin{figure}[!ht]
    \centering
    \begin{tikzpicture}[scale=0.8]
    
    \node[shape=circle,draw=black] (A) at (2,0) {$u$};    
    \node[shape=circle,draw=black] (B) at (2,2) {$v_1$};
    \node[shape=rectangle,draw=white] (B') at (2,2.85) {$T_1$};
    \node[shape=circle,draw=black] (C) at (4,-1) {$v_2$};
    \node[shape=rectangle,draw=white] (C') at (4.9,-1) {$T_2$};
    \node[shape=circle,draw=black] (D) at (0,-1) {$v_3$};
    \node[shape=rectangle,draw=white] (D') at (-0.9,-1) {$T_3$};
    
    \draw(1,1.4) rectangle (3,3.4);
    \draw(3.4,0) rectangle (5.4,-2);
    \draw(-1.4,0) rectangle (0.6,-2);

   \path[->,>=latex, bend left=30](B) edge (A);
   \path[->,>=latex, bend left=30](A) edge (B);
   \path[->,>=latex, bend left=50, red,dashed, thick](A) edge (B);
   \path[->,>=latex, bend left=30](C) edge (A);
   \path[->,>=latex, bend left=50](C) edge (A);
   \path[->,>=latex, bend left=30](A) edge (C);
   \path[->,>=latex, bend left=30](D) edge (A);
   \path[->,>=latex, bend left=30](A) edge (D);
   \path[->,>=latex, bend left=50](A) edge (D);

    \draw [ thick,->,>=stealth, red](4.5,2) arc (0:330:0.4cm);

\end{tikzpicture}
\caption{We sketch a stopping tree-like rotor  multigraph as follows: a vertex $u$, its neighbours $v_1,v_2,v_3$ (that might be sinks), respectively belonging to $T_1$, $T_2$,  $T_3$, the three connected components of $\bar{G} \setminus\{u\}$. In particular, we have $\arcs^+(u)=\{(u,v_1); (u,v_2); (u,v_3)\}$. We consider the rotor configuration in red on $u$ and $\theta_{u}$ is the anticlockwise order on the arcs of ${\cal A}^+(u)$.
Consider the routing of a particle starting at $u$:
\begin{itemize}
    \item the particle moves from $u$ to $v_1$, and stays for a while in the subtree $T_1$ -- where it either reaches a sink or comes back to $u$. Suppose it comes back to $u$. Then:
    \item the particle moves from $u$ to $v_3$, and either reaches a sink in $T_3$ or comes back to $u$. Suppose it comes back to $u$ once again;
    \item the rotor walk goes on, in $T_3$, then in $T_2, T_1, T_1, T_3, \dots$
    \item until finally the particle ends in a sink in one of the subtrees, say $T_2$.
\end{itemize}
Now consider only the relative movement that the particle had in $T_2$: it went from $u$ into $T_2$ and back to $u$ a number of times, before it ended in a sink. If we were to replace $T_1$ and $T_3$ by single arc leading back automatically to $u$, the relative movement in $T_2$ would have been exactly the same. The return flow will be a quantity that counts exactly the ability of each subtree to bounce back the particle to $u$. During the process described above, every time the particle enters a subtree and comes back to $u$, we can think of it as consuming a single unit of return flow in this subtree. The first time that a particle enters a subtree that has exactly one unit of return flow left, then the particle must end in a sink of that subtree.}

\label{IntuitionRN}
\end{figure}
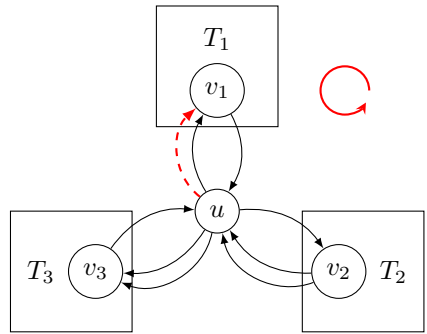

\begin{definition}[$(u,v)$-subtree]
\label{DefSubtree}
Let $(u,v) \in \arcs$. The
$(u,v)$-subtree $T_{(u,v)}$ is a sub(multi)graph of $G$:
\begin{itemize}
    \item whose vertices are all the vertices of the connected component of $\bar{G} \setminus \{u\}$ that contains $v$, together with $u$;
    \item whose arcs are all the arcs of $G$ that link the vertices above, excepted in
    $u$ where we remove all arcs of ${\cal A}^+(u)$ but a single arc $a$ with head $v$. Such an arc $a$ always exists because $(u,v) \in \arcs$;
    \item whose rotor orders are unchanged except at $u$ where $\theta_u(a)=a$.
\end{itemize}
\end{definition}
 Such a subtree is a (not necessarily stopping) tree-like rotor multigraph. A rotor configuration $\rho$ in $G$ can be thought of as
 a rotor configuration $\rho'$ in $T_{(u,v)}$ by defining that $\rho'(u) = a$ and $\rho'(w) = \rho(w)$ for all $w \in T_{(u,v)}$.

We define a notion of \emph{flow} for a particular starting vertex.

\begin{definition}[Flow of $(u,v)$]
We define the \emph{flow} on arc $(u,v) \in \arcs$ for  configuration $\rho$, denoted by $F_{\rho}(u,v)$, the number of times (possibly infinite) that an arc with tail $u$ and head $v$ is visited during the maximal rotor walk of a particle starting from the rotor-particle configuration $(\rho,u)$. We denote by $F_\rho(u)$ the flow vector of $(u,v)$ for every $v \in \Gamma^+(u)$.
\end{definition}

\begin{definition}[Return flow]
\label{DefRN}
The \emph{return flow} of arc $(u,v) \in {\cal \hat{A}}$ for configuration $\rho$, denoted by $\rf{u}{v}{\rho}$, is the flow on $(u,v)$ in the $(u,v)$-subtree $T_{(u,v)}$.
\end{definition}

Note that the return flow $\rf{u}{v}{\rho}$ also corresponds to the number of times the particle visits $u$ while processing a maximal rotor walk from $(\rho,u)$ in $T_{(u,v)}$ (see \autoref{fig : calcRN}). 
By definition of return flow, if $u \in S_0$, then  $\rf{u}{v}{\rho}=0$, and if $v \in S_0$, or if $(v,u) \notin {\cal \hat{A}}$ then $\rf{u}{v}{\rho} = 1$. 

Remark also that, even if the tree-like multigraph is stopping, it is not necessarily the case of any $(u,v)$-subtree: this is for instance the case of a leaf $v$ which is not a sink such that $(u,v) \in {\cal \hat{A}}$. Finiteness of the return flow characterizes the subtrees that are stopping as stated in \autoref{lem:uv_stopping}.

\begin{lemma}
\label{lem:uv_stopping}
Given a stopping tree-like multigraph $G$
and $(u,v) \in {\cal \hat{A}}$,
the $(u,v)$-subtree $T_{(u,v)}$ is stopping if and only if
for any rotor configuration on $G$, the return flow of $(u,v)$ is finite. 
\end{lemma}
 
 \begin{proof}
 If the $(u,v)$-subtree is stopping, then by \autoref{lem:finiteRotorWalk} any rotor walk is finite, and the return flow is finite.

 If the $(u,v)$-subtree is not stopping, there is a sink component $C$ in $T_{(u,v)}$. If $C$ does not contain $u$, it is also a sink component of $G$ since we do not add or remove any arc in $C$ while going from $T_{(u,v)}$ to $G$. But $G$ is assumed stopping hence such sink component does not exist. The only possibility is that $C$ contains $u$. In $C$, there is a vertex that will be visited infinitely often while processing a maximal rotor walk. But  this will be the case for its neighbours in $C$ as well, and, transitively, for every vertex in $C$. In particular $u$ will be visited infinitely often, hence for any rotor configuration $\rho$, $\rf{u}{v}{\rho}$ is infinite.
 \end{proof}

We give a bound on the maximal value of the return flow in a multigraph as it will be used to express our complexity results later.

\begin{lemma}[Return flow bound]
\label{returnBound}
Let $(u, u_1) \in \arcs$ and $\rho$ be a configuration. Then if there is a directed path  $[u, u_1,\dots,u_k,s]$ from $u$ to a sink $s \in S_0$ then $\rf{u}{u_1}{\rho}$ satisfies
$$\rf{u}{u_1}{\rho} \leq \prod_{i=1}^k|{\cal A}^+(u_i)| \leq e^{|{\cal A}|/e},$$
otherwise $\rf{u}{u_1}{\rho}$ is infinite. In particular this shows that return flows can be written in at most $O(|\cal A|)$ bits.
\end{lemma}

 \begin{proof}
 The first inequality derives straightforwardly from the proof of \autoref{lem:finiteRotorWalk}. A similar bound can be found in Lemma~1 in~\cite{gartner_et_al:LIPIcs.ICALP.2021.69}.
 
 For the second inequality, the previous bound has maximal value if all vertices on the path have $|{\cal A}|/n$ outgoing arcs with $n$ the number of vertices. This is maximal for $n=|{\cal A}|/e$ in which case the bound becomes $e^{|{\cal A}|/e}$. Hence, the return flow can be written in $O(|\cal A|)$ bits.
 \end{proof}

Return flows and flows are linked by the following result:

\begin{lemma}
\label{lem:flowReturnFlow}
Given a stopping tree-like rotor multigraph $G$, 
consider $u \in V_0$ and suppose that $h(D(\rho)(u)) = v$. Then:
\begin{itemize}
    \item $F_\rho(u,v) = \rf{u}{v}{\rho} $ ;
    \item for all $w \in \Gamma^+(u) \setminus \{v\}, \quad F_\rho(u,w) < \rf{u}{w}{\rho}$;
    \item for all $w \in \Gamma^+(u) \cap \Gamma^-(u) \setminus \{v\}, \rf{w}{u}{\rho} = F_\rho(u,w) + 1$.
\end{itemize}
\end{lemma}

\begin{proof}
 Let $\rho'$ be the configuration obtained after routing a particle from $(\rho,u)$ until the particle is on $u$ for the last time, in which case $h(\rho'(u))=v$ and $\rf{u}{v}{\rho'}=1$.
 By definition of $\rho'$, we have $F_{\rho'}(u,v) = 1$ and $F_{\rho'}(u,w) = 0$ for all $w \in \Gamma^+(u) \setminus \{v\}$. If, moreover, $w \in \Gamma^-(u)$ then $\rf{w}{u}{\rho'} = 1 = 1 + F_{\rho'}(u,w)$. Hence $\rf{u}{v}{\rho'} = 1 = F_{\rho'}(u,v)$ and by definition of the return flow, $\rf{u}{w}{\rho'} \geq 1$ so that $\rf{u}{w}{\rho'} > F_{\rho'}(u,w)$. Hence the property is satisfied for $\rho'$. 
 
 On the other hand it should be clear that $F_{\rho}(u,w) - F_{\rho'}(u,w) = \rf{u}{w}{\rho} - \rf{u}{w}{\rho'}$ for every $w \in \Gamma^+(u)$. If, moreover,  $w \in \Gamma^-(u)$ and $w \neq v$ the previous quantity is also equal to $\rf{w}{u}{\rho} - \rf{w}{u}{\rho'}$. Hence the property is true for $\rho$ as well.
\end{proof}

\begin{figure}[!ht]
    \centering
    \begin{tikzpicture}[scale=1]
         
    \node[shape=circle,draw=black] (I) at (2,2) {$u_1$};
    \node[shape=circle,draw=black] (E) at (4,2) {$u_3$};
    \node[shape=circle,draw=black, thick] (J') at (6,2) {$s_0$};
    \node[shape=circle,draw=black] (J) at (6,2) {$~~~~~~$};
    \node[shape=circle,draw=black,thick] (A) at (0,0) {$u_4$};
    \node[shape=circle,draw=black] (B) at (2,0) {$u_0$};
    \node[shape=circle,draw=black] (C) at (4,0) {$u_2$};
    \node[shape=circle,draw=black,thick] (D') at (6,0) {$s_1$};
    \node[shape=circle,draw=black] (D) at (6,0) {$~~~~~~$};
    
   \path [->, >=latex, bend left=30, very thick,dashed, red](B) edge (C);
   \path [->, >=latex, bend left=30, very thick,dashed, red](C) edge (B);
   \path [->, >=latex, bend left=30](D) edge (C);
   \path [->, >=latex, bend left=30](C) edge (D);
   \path [->, >=latex, bend left=30](B) edge (I);
   \path [->, >=latex, bend left=30, very thick,dashed, red](I) edge (B);
   \path[->,>=latex, bend left=30](B) edge (A);
   \path[->,>=latex,very thick,dashed, red, bend left=30](A) edge (B);
   \path[->,>=latex, bend left=30](I) edge (E);
   \path[->,>=latex, bend left=30, very thick,dashed, red](E) edge (I);
   \path[->,>=latex, bend left=30](J) edge (E);
   \path[->,>=latex, bend left=30](E) edge (J);
   
    \node at (3,1.4) {\scriptsize{2}};
    \node at (2.6,1) {\scriptsize{2}};
    \node at (3,2.6) {\scriptsize{2}};
    \node at (1.4,1) {\scriptsize{3}};
    \node at (3,0.6) {\scriptsize{2}};
    \node at (3,-0.6) {\scriptsize{2}};
    \node at (1,0.6) {\scriptsize{2}};
    \node at (1,-0.6) {\scriptsize{$+\infty$}};
    \node at (5,0.6) {\scriptsize{1}};
    \node at (5,-0.6) {\scriptsize{0}};
    \node at (5,2.6) {\scriptsize{1}};
    \node at (5,1.4) {\scriptsize{0}};

\draw [ thick,->,>=stealth, red](1,2) arc (0:330:0.4cm);
\end{tikzpicture}
\caption{Examples of return flows in a simple graph. The rotor configuration is depicted by red arcs in dashes, with $S_0 = \{s_0,s_1\}$, and $\theta_{u_i}$ is the anticlockwise order on every vertex. We write the return flow of all arcs of $\arcs$ next to their corresponding arc in ${\cal A}$.
As a tutorial example, we detail the computation of $\rf{u_1}{u_0}{\rho}$ and $\rf{u_0}{u_1}{\rho}$. In the $(u_1,u_0)$-subtree, the particle will visit the following sequence of vertices $u_1,u_0,u_2,u_0,u_1,u_0,u_4,u_0,u_2,s_1$, where it crosses $(u_1,u_0)$ twice, thus $\rf{u_1}{u_0}{\rho}=2$. For the $(u_0,u_1)$-subtree, the sequence of vertices visited by the particle is $u_0,u_1,u_0,u_1,u_3,u_1,u_0,u_1,u_3,s_0$ hence $\rf{u_0}{u_1}{\rho}=3$.}
\label{fig : calcRN}
\end{figure}
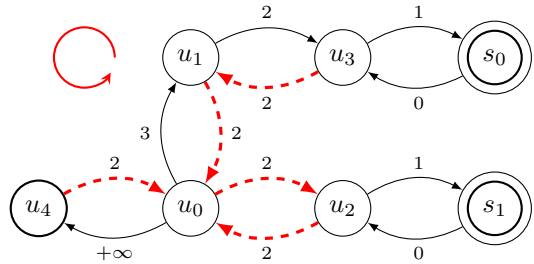

\subsection{Revolving Routine}
\label{subsec:revolvRoutine}

Based on the foregoing, using \autoref{lem:flowReturnFlow}, in order to calculate $D(\rho)(v)$, we need to compute the flow of all arcs $(u, w)$ with $w \in \Gamma^+(u)$ and compare it to $\rf{u}{w}{\rho}$. This idea is introduced in the example drawn in \autoref{ExempleRN}.

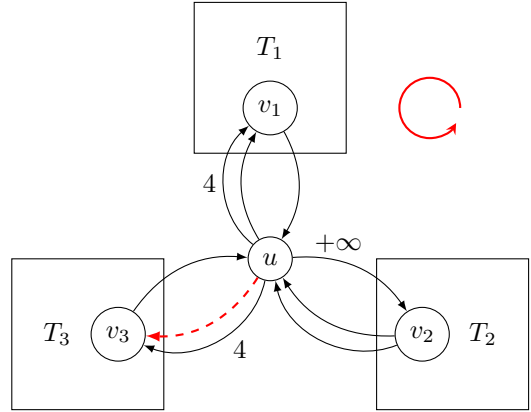
\begin{figure}[!ht]
    \centering
    \begin{tikzpicture}
    
    \node[shape=circle,draw=black] (A) at (2,0) {$u$};    
    \node[shape=circle,draw=black] (B) at (2,2) {$v_1$};
    \node[shape=rectangle,draw=white] (B') at (2,2.8) {$T_1$};
    \node[shape=circle,draw=black] (C) at (4,-1) {$v_2$};
    \node[shape=rectangle,draw=white] (C') at (4.8,-1) {$T_2$};
    \node[shape=circle,draw=black] (D) at (0,-1) {$v_3$};
    \node[shape=rectangle,draw=white] (D') at (-0.8,-1) {$T_3$};
    
    \draw(1,1.4) rectangle (3,3.4);
    \draw(3.4,0) rectangle (5.4,-2);
    \draw(-1.4,0) rectangle (0.6,-2);

   \path[->,>=latex, bend left=30](B) edge (A);
   \path[->,>=latex, bend left=30](A) edge (B);
   \path[->,>=latex, bend left=50](A) edge (B);
   \path[->,>=latex, bend left=30](C) edge (A);
   \path[->,>=latex, bend left=50](C) edge (A);
   \path[->,>=latex, bend left=30](A) edge (C);
   \path[->,>=latex, bend left=30](D) edge (A);
   \path[->,>=latex, bend left=30, red,dashed, thick](A) edge (D);
   \path[->,>=latex, bend left=50](A) edge (D);
   
    \node at (1.2,1) {4};
    \node at (1.6,-1.2) {4};
    \node at (2.9,0.2) {$+\infty$};

    \draw [ thick,->,>=stealth, red](4.5,2) arc (0:330:0.4cm);

\end{tikzpicture}
\caption{Consider the same stopping tree-like rotor  multigraph as in \autoref{IntuitionRN} where $\rho(u)$ is the red arc in dashes. The return flow of all arcs of $\arcs^+(u)$ are given next to their corresponding arcs in ${\cal A}$.}
\label{ExempleRN}
\end{figure}

In \autoref{ExempleRN}, let us put a particle on $u$ and  route it until it comes back to $u$. The first time,  the particle will travel through the red arc in dashes and land on $v_3$. Then, the red arc in dashes (i.e. $\rho(u)$) is updated to $\theta_u(\rho(u))$ which is the other arc with tail $u$ and head $v_3$. As $\rf{u}{v_3}{\rho} >1$, the particle does not reach a sink in $T_3$, and the particle will come back to $u$ after travelling in $T_3$. After this process (the particle's walk in $T_3$), the return flow of $(u,v_3)$ has decreased by one as $(u,v_3)$ has been crossed exactly once.
\newline
During the next step of this rotor walk, the particle will travel to $v_3$ again, make several moves in $T_3$ without reaching a sink and comes back to $u$. The return flow of $(u,v_3)$ is now $2$. Next, the particle will travel to $v_2$ and so on. At some point, the particle will travel through an arc with tail $u$ and head $v_i$ while the return flow of $(u,v_i)$ is one. In this case, the particle will reach a sink in $T_i$ and then will not come back to $u$. Here this condition is first met for $i=3$ when the return flows of $(u,v_1),(u,v_2),(u,v_3)$ respectively are $2,+\infty,1$ the last time the particle is on vertex $u$.

Inspired by this process, we state the following Theorem.

\begin{theorem}
\label{thm:routine}
For any vertex $u \in V_0$, given the return flows of all arcs $(u,v) \in \hat{A}$ with $v \in \Gamma^+(u)$, one can compute $D(\rho)(u)$ and the flow on each arc $(u,v)$ in time \[ O(|\Gamma^+(u)|~\cdot~c(r_{\max},|{\cal A}^+(u)|))\]
 with $r_\text{max}$ being the maximum (finite) value of $\rf{u}{v}{\rho}$ for all $v \in \Gamma^+(u)$.
\end{theorem}

\begin{proof}
\autoref{alg:routine} is a local routine that abstracts the process described in the caption of \autoref{IntuitionRN}. It  computes $D(\rho)(u)$ as $a$ and the flow of all arcs $(u,v)$ as $F$. Notice that, in the routine and its improved version, we consider that  $+\infty - 1$ and $+\infty / a$ with $0<a<+\infty$ are  equal to $+\infty$. Remark also that, since the input multigraph $G$ of the routine is stopping, there is at least one return flow which is finite. However, since the return flow can be exponential as in the example drawn in \autoref{expchain}, \autoref{alg:routine} might run in exponential time.

\RestyleAlgo{ruled}
\SetKwInOut{Input}{input}\SetKwInOut{Output}{output}

\begin{algorithm}[!ht]
\caption{Revolving Routine}\label{alg:routine}
\DontPrintSemicolon
\Input{$u$ is a vertex of $V$; $r :\Gamma^+(u) \rightarrow \mathbb Z \cup \{+\infty\}$ contains the return flows of arcs $(u,v)$ of $\arcs$ with $v \in \Gamma^+(u)$.}
\Output{$a$ an arc of ${\cal A}^+(u)$; $F$ such that $F(v)$ is the number of times an arc of ${\cal A}^+(u)$ with head $v$ has been visited}
\BlankLine
$F(v) \gets 0$ for all $v\in \Gamma^+(u)$\;
$a \gets \rho(u)$\;
\While{$r(h(a)) > 1$}{
 
  $r(h(a)) \leftarrow r(h(a)) - 1$\;
  $F(h(a)) \leftarrow F(h(a)) + 1$\;
  $a \leftarrow \theta_u( a )$\;
}
\Return{$a$, $F$}
\end{algorithm}

We can speed up the computation by noting that every time the rotor at $u$ makes one full turn, we know exactly how many times each value of the flow $F$ has increased. This remark leads to the improved \autoref{alg:ImprovedRoutine} denoted by IRR in the rest of the document. The IRR consists in two steps. First, we compute how many full turns the rotor on $u$ does before the routine ends. Then, we use \autoref{alg:routine}.

\begin{algorithm}[!ht]
\caption{Improved Revolving Routine  (IRR)}\label{alg:ImprovedRoutine}
\DontPrintSemicolon
\Input{$u$ is a vertex of $V$; $r :\Gamma^+(u) \rightarrow \mathbb Z \cup \{+\infty\}$ contains the return flows of  arcs $(u,v)$ of $\arcs$ with $v \in \Gamma^+(u)$.}
\Output{$a$ an arc of ${\cal A}^+(u)$; $F$ such that $F(v)$ is the number of times an arc of ${\cal A}^+(u)$ with head $v$ has been visited}
\BlankLine
For $v \in \Gamma^+(u)$, let $Q(v)$ be the quotient of the euclidean division of $r(v)$ by $|h^{-1}(v)|$; let $q_{\text{min}}$ be the minimum value of $Q(v)$ for all $v$; let $R(v)$ be $r(v)-(q_{\text{min}}*|h^{-1}(v)|)$ and let $F(v)=0$.

\Comment{This step corresponds to the return flow diminutions that occurs during the $q_{\text{min}}*|{\cal A}^+(u)|$ first steps, and it also ensures that there exists at least one $v$ such that in less than $|{\cal A}^+(u)|$ steps, $R(v)<1$.}
\BlankLine
$a \gets \rho(u)$\;

\While{$R(h(a)) > 1$}{
  $R(h(a)) \mathrel{-}=1$\;
  $a \leftarrow \theta_u( a )$\;
}
\For{$v \in \Gamma^+(u)$}{
    $F(v) \gets r(v)-R(v)$\;
    }
\Return{a,F}
\end{algorithm}

At line 1 of IRR, computing $Q(v)$ is done in time $O\big(|\Gamma^+(u)|~\cdot~c(r_{\max},|{\cal A}^+(u)|)\big)$. Then, the loop at line 3 stops in at most $|{\cal A}^+(u)|$ steps, which is small compared to the first term. The same applies for the loop at line 6. In the end, IRR runs in time $O(|\Gamma^+(u)|~\cdot~c(r_{\max},|{\cal A}^+(u)|))$.

\end{proof}

As the routines are crucial for the rest of the article, we now state
an important monotony property. Here, we fix  a vertex $v$ in $G$ and consider on the outgoing arcs of $v$ the respective return flows $r_1$ and $r_2$ of two rotor configurations $\rho_1$ and $\rho_2$.

\begin{lemma}[Monotony of the flow]
\label{lem:monotonyRN} 
Let $u \in V_0$, and  $\rho_1, \rho_2$ be two rotor configurations such that $\rho_1(u) = \rho_2(u)$. Let $v_i = h(D(\rho_i)(u))$ for $i \in \{1,2\}$ and $\overline{\Gamma}$ be the set of vertices $v \in \Gamma^+(u)$ such that $\rf{u}{v}{\rho_1} \geq \rf{u}{v}{\rho_2}$. If $v_1 \in \overline{\Gamma}$, then
\begin{itemize}
    \item $F_{\rho_1}(u) \geq F_{\rho_2}(u)$ component-wise,
    \item $v_2 \in \overline{\Gamma}$.
\end{itemize}
\end{lemma}

\begin{proof}

Let $\rho_i^k$ for $i \in \{1, 2\}$ be the rotor configuration after the particle has visited $u$ exactly $k$ times during a maximal rotor walk starting from $(\rho_i, u)$. In particular, $\rho_i^1 = \rho_i$. We denote by $R_i^k(u, v)$ and $F_i^k(u, v)$ the quantities $\rf{u}{v}{\rho_i^k}$ and $F_{\rho_i^k}(u,v)$ respectively.
Let $K_i$  be the last time the particle is on $u$ in this walk, which is characterized by $R^{K_i}_i(u,v_i) = 1$ and $h(\rho_i^{K_i}(u))=v_i$ for $i \in \{1, 2\}$. 

For all $k \leq \min(K_1, K_2)$, let $\overline{\Gamma}^k$ be the set of vertices $v \in \Gamma^+(u)$ such that $R_1^k(u, v) \geq R_2^k(u, v)$. It turns out that in fact, for all such $k$ we have $\overline{\Gamma}^k = \overline{\Gamma}$ since values of both $R_1$ and $R_2$ are decremented simultaneously. 

We first show that $K_2 \leq K_1$. By contradiction, assume that $K_2 > K_1$. Since $v_1 \in \overline{\Gamma}$, then, at step $K_1$, $R^{K_1}_1(u,v_1) = 1 \geq R^{K_1}_2(u,v_1)$ and $h(\rho_i^{K_i}(u))=v_1$ for $i \in \{1, 2\}$. This implies that $K_1$ is the last time that the walk starting at $(\rho_2, u)$ is at $u$, \emph{i.e.} $K_2 = K_1$, hence a contradiction. 

For every $v \in \Gamma^+(u)$, we have $F_{\rho_i}(u,v) = \Delta F^k_i(u,v) + F^k_i(u, v)$ with $\Delta F^k_i(u,v) = F_{\rho_i}(u,v) - F^k_i(u, v)$.  $\Delta F^k_i(u, v)$ is the number of times each arc $(u,v)$ has been used until step $k$. Since $\rho_1(u) = \rho_2(u)$, and as long as $k \leq K_2$,  $\Delta F^k_i(u, v)$ does not depend on $i$. It follows that $F_{\rho_1}(u,v) - F_{\rho_2}(u,v) = F^{K_2 }_1(u, v) - F^{K_2}_2(u, v) $. If $v \neq v_2$ then $F^{K_2 }_1(u, v) - F^{K_2}_2(u, v) = F^{K_2 }_1(u, v) - 0 \geq 0$. Otherwise $F^{K_2 }_1(u, v_2) - F^{K_2}_2(u, v_2) = F^{K_2 }_1(u, v_2) - 1 $. Since $K_1 \geq K_2$ and $h(\rho_1^{K_2}(u)) = v_2$, arc $(u, v_2)$ will be used at least once more during the walk, \emph{i.e.} $F^{K_2 }_1(u, v_2) \geq 1$. Hence the difference is positive which shows the first part of the lemma.

Let $v$ be such that $v \notin \overline{\Gamma}$. Then $R_2^{K_2}(u, v) > R_1^{K_2}(u, v) \geq 1$. Hence $R_2^{K_2}(u, v) \geq 2$ which implies $v_2 \neq v$ and then $v_2 \in \overline{\Gamma}$.
\end{proof}

\section{ARRIVAL for Tree-like Multigraphs}
\label{sec:ARRIVALZero}

In this section we show that, for a given configuration $\rho$, we can compute the Destination Forest $D(\rho)$ in time complexity $O(|{\cal A}|
 \cdot c(r_{\max}, |{\cal A}|))$, hence solve the ARRIVAL problem for every vertex at the same time. To achieve this, we recursively compute return flows for all arcs in $\arcs$ and then use these flows to compute the destination forest.
 
 In all this section, let $G=(V_0,S_0,{\cal A},h,t,\theta)$ be a stopping tree-like rotor multigraph and $\rho$ be a rotor configuration on $G$.

The next two lemmas show how to compute the return flows by using \autoref{thm:routine}.

\begin{lemma}
\label{lem:propag_ret}
For any two vertices $u$ and $v$ such that $(u,v) \in \arcs$, and given $\rf{v}{w}{\rho}$ for every $w \in \Gamma^+(v) \setminus{\{u\}}$, the return flow $\rf{u}{v}{\rho}$ can be computed in time $O(|\Gamma^+(v)| \cdot c(r_{\max},|{\cal A}^+(v)|))$. We illustrate this operation in \autoref{fig:ReturnNumberPropag}.
\end{lemma}

\begin{proof}
If $(v,u) \not\in \arcs$, then $\rf{u}{v}{\rho} = 1$. 

Otherwise, if $\rf{v}{w}{\rho} = +\infty$ for all $w \in \Gamma^+(v)$ such that $w \neq u$, then $\rf{u}{v}{\rho} = +\infty$.

In all other cases, apply the IRR to the vertex $v$, with input values $\rf{v}{w}{\rho}$ for all $w \in \Gamma^+(v)$ such that $w \neq u$, and with  $\rf{v}{u}{\rho} = p$, where $p$ is intended to be an integer large enough so that the output $a$ of the IRR is such that $h(a) \neq u$. Then by \autoref{lem:flowReturnFlow}, $\rf{u}{v}{\rho} = F_{\rho}(v,u)+1$, where $F_\rho(v)$ is obtained by the IRR.
\end{proof}

Parameter $p$ in the previous proof should be chosen large enough so that variable
$R(u)$ in the routine remains strictly positive; for instance $p$ can be initialized with $(q_\text{min}+1) \cdot |h^{-1}(u)|$ with $q_\text{min}$ defined in \autoref{alg:ImprovedRoutine}.

\begin{figure}[!ht]
    \centering
    \begin{tikzpicture}
    
    \node[shape=circle,draw=black] (A) at (2,0) {$u$};    
    \node[shape=circle,draw=black] (B) at (2,2) {$v_1$};
    \node[shape=circle,draw=black] (E) at (2,-2) {$v_3$};
    \node[shape=circle,draw=black] (C) at (4,-1) {$v_4$};
    \node[shape=circle,draw=black] (D) at (0,-1) {$v_2$};
    
   \path[->,>=latex, bend left=20, very thick](B) edge (A);
   \path[->,>=latex, bend left=20](A) edge (B);
   \path[->,>=latex, bend left=20](C) edge (A);
   \path[->,>=latex, bend left=20, dashed](A) edge (C);
   \path[->,>=latex, bend left=20](D) edge (A);
   \path[->,>=latex, bend left=20, dashed](A) edge (D);
   \path[->,>=latex, bend left=20](E) edge (A);
   \path[->,>=latex, bend left=20, dashed](A) edge (E);
   
   \node at (2.5,1) {{r}};
   \node at (1.3,-0.9) {\scriptsize{$\bm{r_a}$}};
   \node at (2.5,-1) {\scriptsize{$\bm{r_b}$}};
   \node at (3.2,-0.08) {\scriptsize{$\bm{r_c}$}};

\end{tikzpicture}
\caption{In this Figure, we illustrate which value is computed with \autoref{lem:propag_ret}. If the return flows $r_a,r_b,r_c$ are known, we can compute the return flow $r$.}
\label{fig:ReturnNumberPropag}.
\end{figure}
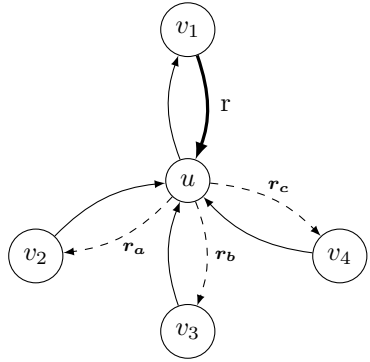

\begin{lemma}
\label{lem: Retropropag}
For any vertex $u \in V_0$, and given $\rf{u}{v}{\rho}$ for every $v\in \Gamma^+(u)$, one can compute the return flow of all arcs $(w,u)$ with $w \in \Gamma^-(u)$ in time  $O(|\Gamma^+(u)| \cdot c(r_{\max},|{\cal A}^+(u)|))$. We illustrate this operation in \autoref{fig:RetroPropag}.
\end{lemma}

\begin{proof}
For all $w \in \Gamma^-(u) \setminus{\Gamma^+(u)}$, we have $\rf{w}{u}{\rho} = 1$.

We use \autoref{thm:routine} once on $u$ to compute $D(\rho)(u)$ and the vector flow $F_{\rho(u)}$. Let $v = h(D(\rho)(u))$.
Then, by \autoref{lem:flowReturnFlow}, for all $w \in \Gamma^-(u) \cap \Gamma^+(u) \setminus{v}$, we have $\rf{w}{u}{\rho} = F_{\rho}(u,w) +1$.
It remains to apply \autoref{lem:propag_ret} once more to compute $\rf{v}{u}{\rho}$. All in all, we have used \autoref{thm:routine} twice, hence the complexity. 
\end{proof}

One can check by applying \autoref{lem: Retropropag} that, on the example of \autoref{ExempleRN}, we have $\rf{v_1}{u}{\rho}=3$, $\rf{v_2}{u}{\rho}=2$ and $\rf{v_3}{u}{\rho} = 4$.

\begin{figure}[!ht]
    \centering
    \begin{tikzpicture}
    
    \node[shape=circle,draw=black] (A) at (2,0) {$u$};    
    \node[shape=circle,draw=black] (B) at (2,2) {$v_1$};
    \node[shape=circle,draw=black] (E) at (2,-2) {$v_3$};
    \node[shape=circle,draw=black] (C) at (4,-1) {$v_4$};
    \node[shape=circle,draw=black] (D) at (0,-1) {$v_2$};
    
   \path[->,>=latex, bend left=20, very thick](B) edge (A);
   \path[->,>=latex, bend left=20, dashed](A) edge (B);
   \path[->,>=latex, bend left=20,very thick](C) edge (A);
   \path[->,>=latex, bend left=20, dashed](A) edge (C);
   \path[->,>=latex, bend left=20, very thick](D) edge (A);
   \path[->,>=latex, bend left=20, dashed](A) edge (D);
   \path[->,>=latex, bend left=20, very thick](E) edge (A);
   \path[->,>=latex, bend left=20, dashed](A) edge (E);
   
   \node at (1.5,1) {\scriptsize{$\bm{r_d}$}};
   \node at (1.3,-0.9) {\scriptsize{$\bm{r_a}$}};
   \node at (2.5,-1.2) {\scriptsize{$\bm{r_b}$}};
   \node at (3.2,-0.08) {\scriptsize{$\bm{r_c}$}};
   \node at (2.5,1) {\scriptsize{$\bm{r_1}$}};
   \node at (0.8,-0.08) {\scriptsize{$\bm{r_2}$}};
   \node at (1.5,-1.2) {\scriptsize{$\bm{r_3}$}};
   \node at (2.7,-0.9) {\scriptsize{$\bm{r_4}$}};

\end{tikzpicture}
\caption{In this Figure, we illustrate which values is computed with \autoref{lem: Retropropag}. If the return flows $r_a,r_b,r_c, r_d$ are known, we can compute the return flows $r_1,r_2,r_3,r_4$.}
\label{fig:RetroPropag}.
\end{figure}
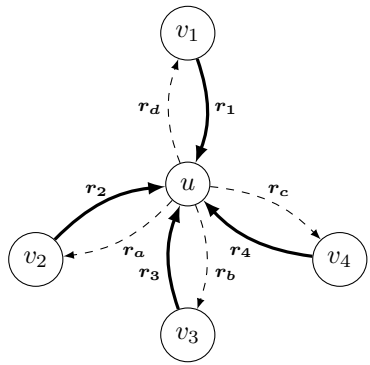
\smallskip

We are now ready to state our main theorem. Complexity bounds are given in two different contexts: 
\begin{itemize}
    \item a context where the time needed for arithmetic computation matters, as in a Turing machine. These bounds are polynomial in the size of the graph and use notation $c(a,b)$ to denote the time complexity of dividing the number $a$ by the number $b$;
    \item another context where  arithmetic operations can be done in constant time, where we achieve linear complexity in the size of the graph.
\end{itemize}

\begin{theorem}[Complete Destination Algorithm]
\label{CDA}
The configuration $D(\rho)$ can be computed in time $O(|{\cal A}|)$ for a stopping tree-like multigraph in a model where arithmetic operations can be made in constant time, or alternatively in $O(|{\cal A}|\cdot c(r_{\max}, |{\cal A}|))$ on a Turing Machine.
\end{theorem}

\begin{proof}
Consider an arbitrary vertex $x$. We proceed to a Breadth-First Search (BFS) starting from $x$ in $\overline{G}$, the simple undirected graph associated with $G$. Let $e_1,e_2, \dots,e_m$ (resp. $u_0,u_1,\dots,u_k$ with $u_0=x$) be the prefix order on the edges (resp. on the vertices)  of $\overline{G}$ obtained during the BFS,  i.e. the order in which the edges (resp. the vertices) are visited. 

The algorithm is split into two phases:
\begin{enumerate}
    \item Computation of return flows for all arcs directed from $x$ towards the leaves: for $t=m,m-1,\dots,2,1$, if the edge $e_t$ corresponds to an existing arc $(u_i,u_j)$ of $\arcs$, such that $(u_i,u_j)$ is directed from $x$ towards a leaf of $G$, consider two cases. Firstly, if $u_j$ is a leaf, then store that either $\rf{u_i}{u_j}{\rho} = 1$ or $\rf{u_i}{u_j}{\rho} = + \infty$ depending on whether $u_j \in S_0$ or not. Secondly, if $u_j$ is not a leaf, then by definition of a BFS and a prefix order, if an arc $(u_j,v)$  with $v \neq u_i$ corresponds to an edge $e_{t'}$, then $t<t'$. Hence, we already know the value of $\rf{u_j}{v}{\rho}$ for every $v \neq u_i$ with $v \in \Gamma^+(u_j)$. This means that we can compute recursively $\rf{u_i}{u_j}{\rho}$ by \autoref{lem:propag_ret}.
   
    \item Computation of return flows of all arcs directed from the leaves towards $x$: when this phase begins, for any vertex $u$ except $x$, all return flows $\rf{u}{v}{\rho}$ for all $v \in \Gamma^+(v)$ are known, excepted the return flow of the arc directed from $u$ to $x$. We use \autoref{lem: Retropropag} applied to vertices $u_i$ in increasing order on $i$ as it guarantees that the conditions to apply \autoref{lem: Retropropag} are met. Furthermore, we also compute $D(\rho)(u_i)$ at the same time.
\end{enumerate}

The time needed for the BFS part is $O(|\arcs| + |V|)$.

In the first phase, we use \autoref{thm:routine} at most once for every vertex $v$, for
the arc $(u,v)$ coming from $x$ towards the leaves.
During the second phase, we we use \autoref{thm:routine} at most twice for each vertex.
All in all, we use \autoref{thm:routine} three times for each vertex. Hence the time complexity is $O(\sum_{v \in V}{(|\arcs^+(v)| \cdot c(r_{\max}, |{\cal A}|)))}$ which amounts to $O(|{\cal A}|\cdot c(r_{\max}, |{\cal A}|))$.
\end{proof}

We showed in \autoref{returnBound} that return flows could be written in at most $O(|{\cal A}|)$ bits which gives an upper bound for $c(r_{\max}, |{\cal A}|)$ of $k |{\cal A}|\log(|{\cal A}|)$ for some constant $k > 0$. It is proved in~\cite{harvey2021integer} that the multiplication of two $n$ bits integers can be done in time $O(n\log(n))$ and as the complexity of the division is equivalent to the complexity of multiplication (see~\cite{brent2010modern}), the bound follows. Thus the complexity of our algorithm is $O(|{\cal A}|^2\log(|{\cal A}|))$ in this context.

\begin{comment}

\section*{ARRIVAL with players}

 Problem ARRIVAL can be seen as a zero-player game where the winning condition it that the particle reaches a particular sink (or set of sinks). The one and two players variants of ARRIVAL (i.e. deterministic analogs of \emph{Markov decision processes} and \emph{Stochastic games}) we are interested in are defined in~\cite{fearnley2017reachability}. AJOUTER LA PARTICULARITE SURLES STRATEGIES QU4ON CONSIDERE In the remaining of this paper, we address those two variants and their possible simplifications if the graph is simple.

To formally describe the set of strategies allowed in tree-like rotor games, we need the following definition.

\begin{definition}[Partial Configuration]
Let $V'$ be a subset of $V_0$, a \emph{partial rotor configuration} on $V'$ is a mapping $\rho'$ from $V'$ to ${\cal A}$ such that $\rho'(v) \in {\cal A}^+(v)$ for all $v \in V'$.
\end{definition}

\begin{definition}[Strategy]
A \emph{strategy} for a player is a partial rotor configuration on a subset of $V_0$.
\end{definition}

\end{comment}

\section{One-player Rotor Game}
\label{sec:ARRIVAL1P}

Problem ARRIVAL can be seen as a zero-player game where the winning condition is that the particle reaches a particular sink (or set of sinks). The one and two players variants of ARRIVAL (i.e. deterministic analogs of \emph{Markov decision processes} and \emph{Stochastic games}) we address in the next sections are inspired from~\cite{fearnley2017reachability}, but differ by the choice of the set of strategies (see the discussion hereafter).

In this section, we specifically consider a game with a single player that controls a subset of vertices $ V_{\Max}$ of $V_0$. Given a rotor configuration on the rest of the vertices of $V_0$, a starting vertex and an integer value for each sink, his goal is to wisely choose the initial rotor configuration of the vertices he controls (his strategy) such that the particle reaches one of the sinks with maximal value. 

A remark is in order here: in the seminal paper~\cite{fearnley2017reachability}, a strategy is defined in a more general way since it consists in choosing an outgoing-arc each time the particle is on a vertex controlled by the player. In particular, for a given vertex, the sequence of arcs may not follow a rotor rule, and the number of strategies is even unbounded. It has been shown in that paper that solving such game is NP-complete. On the one hand, the given reduction of 3-SAT can easily be adapted to fit to our framework showing that our definition of the game, although simpler since the set of strategies is finite, still leads to an NP-complete problem. On the other hand, our results extend naturally to general strategies, but at the cost of more technicalities. For instance, the use of general strategies may lead to non-stopping rotor graphs even if every vertex is connected to a sink. This case also seems to us a very natural extension of the zero player case.

To formally define the game, we introduce the following definition.
\begin{definition}[Partial Configuration]
Let $V'$ be a subset of $V_0$, a \emph{partial rotor configuration} on $V'$ is a mapping $\rho'$ from $V'$ to ${\cal A}$ such that $\rho'(u) \in {\cal A}^+(u)$ for all $u \in V'$.
\end{definition}

A one-player rotor game (resp. one-player tree-like rotor game) is given by \newline $(V_r, V_{\Max}, S_0, {\cal A}, h, t, \theta, \val, \rho)$
where $V_r$, $V_{\Max}$ and $S_0$ are disjoint sets of vertices, such that

\begin{itemize}
    \item $(V_0 , S_0, {\cal A}, h, t, \theta)$ is a rotor graph (resp. tree-like rotor graph) with $V_0=V_r \cup V_{\Max}$;
    \item $\val$ is a map from $S_0$ to $\mathbb{N}$ corresponding to a utility of the player who wants the particle to reach a sink $s$ with the highest possible value $\val(s)$;
    \item $\rho$ is a partial configuration on $V_r$, the initial configuration on the vertices not controlled by the player.
\end{itemize}

The tree-like rotor game is \emph{stopping} if and only if the induced rotor graph $(V_0,S_0,{\cal A},h,t,\theta_v)$ is stopping.

The player is called $\Max$, and a \emph{strategy} for $\Max$ is a partial rotor configuration on $V_{\Max}$.
We denote by $\Sigma_{\Max}$ the finite set of strategies for this player.

Consider a partial rotor configuration $\rho$ on $V_r$ together with strategy $\sigma$ and denote by $(\rho, \sigma)$ the rotor configuration where we apply the partial configuration $\rho$ or $\sigma$ depending
on whether the vertex is in $V_r$ or $ V_{\Max}$.

The {\it value of the game} for strategy $\sigma$ and starting vertex $u_0$ is denoted by $\val_{\sigma}(u_0)$ and is equal to $\val(s)$ where $s$ is the sink reached by a maximal rotor walk from the rotor particle configuration $((\rho, \sigma),u_0)$ if any, and $0$ otherwise. 
As in the zero-player framework,
up to computing strongly connected components that do not contain sinks and replacing each of them with a sink of value $0$, we can suppose that the tree-like rotor game is stopping. In the following, all rotor games we consider are tree-like and stopping unless stated otherwise.

When $u_0$ is fixed, the maximal value of $\val_{\sigma}(u_0)$ over all strategies $\sigma \in \Sigma_{\Max}$ is called the \emph{optimal value} of the game with starting vertex $u_0$ and is denoted by $\val^*(u_0)$. Any strategy $\sigma \in \Sigma_{\Max}$ such that $\val_{\sigma}(u_0)=\val^*(u_0)$ is called an \emph{optimal strategy} for the game starting in $u_0$. Observe that optimal strategies may depend on the choice of $u_0$ as illustrated in \autoref{simpleStrategy}. The one-player ARRIVAL problem consists in computing the optimal value of a given starting vertex in a one-player rotor game.

\begin{figure}[t!]
    \centering
    \begin{tikzpicture}[scale=1]
         
    \node[shape=rectangle,draw=black] (A) at (0,0) {$g$};
    \node[shape=circle,draw=black] (B) at (2,0) {$u$};
    \node[shape=circle,draw=black,thick] (C') at (4,0) {$1$};
    \node[shape=circle,draw=black] (C) at (4,0) {$~~~~~$};
    \node[shape=circle,draw=black] (D) at (-2,0) {$v$};
    \node[shape=circle,draw=black,thick] (E') at (-4,0) {$1$};
    \node[shape=circle,draw=black] (E) at (-4,0) {$~~~~~$};
    \node[shape=circle,draw=black,thick] (F') at (0,2) {$0$};
    \node[shape=circle,draw=black] (F) at (0,2) {$~~~~~$};
    \node[shape=circle,draw=black,thick] (G) at (0,-2) {$0$};
    \node[shape=circle,draw=black] (G) at (0,-2) {$~~~~~$};
    
    \path [->,bend right=20,thick,dashed, red](D) edge (A);
    \path [->,bend right=20,thick, dashed,red](B) edge (A);
    \path [->,bend right=20](A) edge (D);
    \path [->,bend right=20](A) edge (B);
    \path [->](A) edge (F);
    \path [->](A) edge (G);
    
    \path [->](D) edge (E);
    \path [->](B) edge (C);

   \draw [ thick,->,>=stealth, red](2.5,2) arc (0:330:0.4cm);

\end{tikzpicture}
\caption{Simple graph where the optimal strategy depends on the starting vertex $u_0$, with $V_{\Max}=\{g\}$, with $u,v \in V_0$ and with all other vertices being sinks. As in previous examples, the starting configuration is depicted by red arcs in dashes, and the rotor order on all vertices is an anticlockwise order on their outgoing arcs. In the case $u_0=v$, the only optimal strategy is $\sigma(g)=(g,v)$ and the game has value $1$. In the case $u_0=u$, the only optimal strategy is $\sigma(g)=(g,u)$ and the game has value $1$.}
\label{simpleStrategy}
\end{figure}
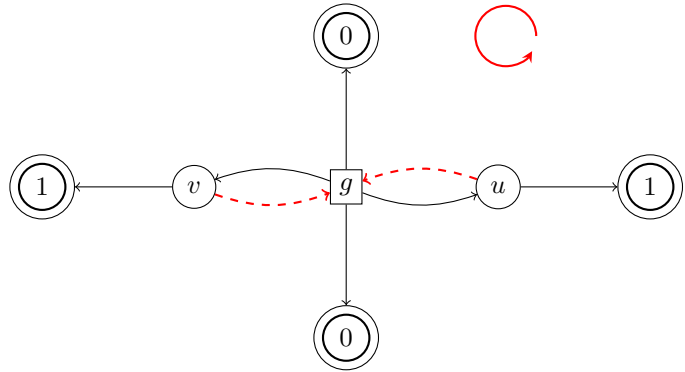

First, we study the case where the values of the sinks are binary, then we adapt those results to the case of nonnegative integer values, and finally we present our results for some different set of strategies.

\subsection{One-player Binary Rotor Game}

In this subsection, we restrict the game to the case where values of sinks are binary numbers i.e. $\val(s) \in \{0,1\}$ for all $s \in S_0$. 
Recall that in the tree-like rotor graph, $T_{(u, v)}$ denotes the $(u,v)$-subtree. We extend this notation to denote the one-player, not necessarily stopping, game played on the $(u,v)$-subtree where we restrict $V_{\Max}$ and $V_r$ to the subtree. For this game, we only consider the case where the starting vertex is $u$.

\begin{definition}[Value under strategy]
Let $(u,v)$ be an arc of $\arcs$. Given a strategy $\sigma$ for the $(u,v)$-subtree, we denote by $\val_{\sigma}(u, v)$ (resp. $\val^*(u, v)$) the value of the game under strategy $\sigma$ (resp. under an optimal strategy) in $T_{(u,v)}$. This is called the value (resp. the optimal value) of the arc $(u,v)$ for strategy~$\sigma$. 
\end{definition}

\begin{definition}[Optimal return flow $r^*$]
Let $(u,v)$ be an arc of $\arcs$. If $\val^*(u, v)=0$, then $\ropti{u}{v}$ is defined as the maximum of $\rf{u}{v}{\sigma}$ over all strategies $\sigma$ on $T_{(u,v)}$, otherwise it is the minimum of $\rf{u}{v}{\sigma^*}$ among \emph{optimal} strategies $\sigma^*$ on $T_{(u,v)}$.
\end{definition}

The next lemma connects the value $\val_\sigma(u)$ with the value of the last outgoing arc of vertex $u$ while processing a maximal rotor walk from the rotor-particle configuration $((\rho,\sigma),u)$.

\begin{lemma}
\label{lem:valFromDT}
Let $a$ be an arc of ${\cal A}^+(u)$ such that $D(\rho,\sigma)(u)=a$ with $h(a)=v$. We have $\val_\sigma(u)=\val_\sigma(u,v)$.  
\end{lemma}

To recursively compute an optimal strategy, we need a stronger notion of optimality, namely a \emph{subtree optimal strategy}.

\begin{definition}[Subtree optimal strategy]
A strategy $\sigma^*$ is \emph{subtree optimal} at $u_0$ if it is optimal at $u_0$ and, moreover, $\val_{\sigma^*}(u, v) = \val^*(u, v)$ and $\rf{u}{v}{\sigma^*} = \ropti{u}{v}$ for every $(u,v)$-subtree such that $(u, v)$ is directed from $u_0$ towards the leaves.
\end{definition}

Instead of recursively computing only the return flow as in the zero-player game, we now propagate both the optimal value and the optimal return flow to construct a subtree optimal strategy. Here, we give an equivalent to \autoref{lem:propag_ret} for the one-player game that details how to recursively compute $\val^*$ and $r^*$.

\begin{lemma}
\label{lem:propag_val}
Let $(u,v) \in \arcs$. 
For every $(v,w)$-subtree with $w \in \Gamma^+(v) \setminus \{u\}$, assume that there is a strategy $\sigma^*_w $ that is subtree optimal. 
Let  $\sigma_v$ be a strategy on the $(u,v)$-subtree $T_{(u,v)}$ such that  $\sigma_v(z) = \sigma^*_w(z)$ when $z \in T_{(v,w)} \cap V_{\Max}$ and $z \neq v$. Furthermore, if $v \in V_{\Max}$, let $\Sigma_v$ be the set of strategies defined on $T_{(u,v)}$ that agree with $\sigma_v$ on every vertex  but $v$. We consider two cases:
\begin{itemize}
    \item if $\val^*_{(u,v)} = 0$, $\sigma_v(v)$ is a strategy in $\Sigma_v$ that maximizes the return flow on $(u,v)$;
    \item if $\val^*_{(u,v)} = 1$, $\sigma_v(v)$ is a strategy in $\Sigma_v$ that is optimal and minimizes the return flow on $(u,v)$.
\end{itemize}
Then $\sigma_v$ is  subtree optimal on $T_{(u,v)}$.
\end{lemma}

\begin{proof}
We suppose that $v \in V_{\Max}$. The case $v \notin V_{\Max}$ can be treated similarly and is omitted.

By assumption, the restriction of $\sigma_v$ to every subtree $T_{(w,z)}$, where $(w,z)$ is an arc of $T_{(u,v)}$ different from $(u,v)$, and directed from $u$ towards the leaves, is subtree optimal. It remains to show that $\val_{\sigma_v}(u,v) = \val^*(u,v)$ and  $\rf{u}{v}{\sigma_v} = \ropti{u}{v}$.
\begin{itemize}
    \item Assume that $\val^*(u,v) = 0$. In this case, we have $\val_{\sigma_v}(u,v) = \val^*(u,v) = 0$ as for any strategy. 
    
    For the return flow, we consider different  cases. 
    
    If $(v,u) \notin \arcs$, then $\rf{u}{v}{\sigma} = 1$ for every strategy $\sigma$ on $T_{(u,v)}$, so $\rf{u}{v}{\sigma_v} = 1 = \ropti{u}{v}$ is maximal. 
    If $(v,u) \in \arcs$ and $\ropti{v}{w} = +\infty$ for every $w \in \Gamma^+(v) \setminus \{u\}$, then $\rf{u}{v}{\sigma} = +\infty$ for every strategy $\sigma$ on $T_{(u,v)}$, in particular $\rf{u}{v}{\sigma_v} = +\infty = \ropti{u}{v}$. 
    
    Finally, if $(v,u) \in \arcs$ and there is $w$ such that $\ropti{v}{w} < +\infty$,
    let $w_0 = h(D(\rho,\sigma_v)(v))$ (defined in the stopping rotor graph $T_{(u,v)}$) and consider a strategy $\sigma$ defined on $T_{(u,v)}$ with $\sigma(v) = \sigma_v(v)$. By \autoref{lem:valFromDT}, we have $\val_{\sigma_v}(v) = \val_{\sigma_v}(v, w_0)$ hence $\val_{\sigma_v}(v, w_0) = 0$. Since $\sigma_v$ is subtree optimal on the $(v, w_0)$-subtree, it follows that $\val^*(v, w_0) = 0$ and then $\val_\sigma(v, w_0) = 0$ 
    and finally we have $\rf{v}{w_0}{\sigma} \leq \rf{v}{w_0}{\sigma_v}$.

    We have $\sigma(v) = \sigma_v(v)$ and $w_0$ is such that $\rf{v}{w_0}{\sigma} \leq \rf{v}{w_0}{\sigma_v}$. Following \autoref{lem:monotonyRN} applied to $v$, $w_0 \in \overline{\Gamma}$, and then $F_{\sigma}(v,u) \leq F_{\sigma_v}(v,u)$.  By \autoref{lem:flowReturnFlow}, we have $\rf{u}{v}{\sigma} = F_{\sigma}(v,u) +1 $ and $\rf{u}{v}{\sigma_v} = F_{\sigma_v}(v,u) +1 $. This implies that $\rf{u}{v}{\sigma} \leq \rf{u}{v}{\sigma_v}$. Since 
     $\sigma_v(v)$ is chosen so that to maximize the return flow on the set of strategies $\Sigma_v$ the result follows.

    \item Assume that $\val^*(u,v) = 1$. We first show that there is an optimal strategy in $\Sigma_v$.
    
    For this, consider an optimal strategy $\sigma^*$ on $T_{(u,v)}$ and let $w_0 = h(D(\rho,\sigma^*)(v))$. By \autoref{lem:valFromDT}, it follows that $\val_{\sigma^*}(v, w_0) = 1$ and then $\val^*(v, w_0) = 1$. Let $\sigma$ be the strategy in $\Sigma_v$ such that $\sigma(v) = \sigma^*(v)$. Since $\sigma$ is subtree optimal on $T_{(v,w_0)}$ we have $\rf{v}{w_0}{\sigma^*} \geq \rf{v}{w_0}{\sigma}$. On the other hand, let $W$ be the set of vertices $w\in \Gamma^+(v) \setminus \{u\}$ such that  $\val^*(v, w) = 0$. On this set $\rf{v}{w}{\sigma^*} \leq \rf{v}{w}{\sigma}$, \emph{i.e.} $w \notin \overline{\Gamma}$ following the notation of \autoref{lem:monotonyRN}. Hence $D(\rho,\sigma)(v) \notin W$ which implies $\val_{\sigma}(u,v) = 1$.
    
    Now, showing that $\rf{u}{v}{\sigma_v} = \ropti{u}{v}$ is done exactly the same way as the case $\val^*(u,v) = 0$. 
    
\end{itemize}
\end{proof}

Note that if $v \in S_0$, then there is no decision to make in $T_{(u,v)}$, and the empty strategy is subtree optimal. Otherwise, \autoref{lem:propag_val} shows inductively that such subtree optimal strategy exists for any $T_{(u,v)}$ where $(u,v) \in \arcs$.

As a second remark, \autoref{lem:propag_val}  can straightforwardly be adapted to the case where the player seeks to minimize the value, by swapping the role of subtrees of value 0 and 1. This will be used in next section when we consider a two-player game.

However, this process requires to determine $\sigma_v(v)$ which minimizes or maximizes (depending of the optimal value of arc $(u,v)$) the return flow. To avoid an additional $|A^+(v)|$ factor in the time complexity by trying all possible choices for $\sigma_v(v)$, we propose \autoref{alg:OptimalStratBin} which runs in time $O(|\Gamma^+(v)| \cdot  c(r_{\max}, |{\cal A}|))$.
In this algorithm, $a_s$ is an arc that will try all possible starting configurations on $v$ in the IRR ; whereas $a_e$ is the corresponding output arc, i.e. the destination arc if we start in $a_s$. The important fact here is that, when $a_s$ in incremented by $\theta_v$, $a_e$ possibly moves in the cyclic ordering but can never make a full turn and go beyond $a_s$; this is because the loop part of the IRR never makes a full turn. During this process, we just keep track of the maximum and minimum return flows depending on the value of the subtree in the direction $a_e$.

Now, in the same spirit than for the zero-player game, we can use \autoref{lem:propag_val} as the basis of a recursive algorithm for computing $\val^*(u,v)$ and $\ropti{u}{v}$ for all $(u,v) \in \arcs$ directed from $v_0$ towards the leaves. This leads to our main theorem.

\begin{theorem}[Computation of $\val^*(u_0)$]
\label{thm:opt_complex}
The optimal value $\val^*(u_0)$ can be computed in the same time complexity as the computation of $D(\rho)$ in the zero-player game (see \autoref{CDA}).
\end{theorem}

\begin{proof}
    By using \autoref{lem:propag_val}, we recursively compute $\val^*(u,v)$ and $\ropti{u}{v}$ for all arcs $(u,v) \in \arcs$ such that $(u,v)$ is directed from $u_0$ towards the leaves. Several cases are considered:
    \begin{itemize}
        \item if $v \in S_0$ then $\val^*(u, v) = \val(v)$ and  $\ropti{u}{v} = 1$;
        \item otherwise we run \autoref{alg:OptimalStratBin}.
    \end{itemize}
    This process is done at most once for each vertex $v$ which results in the time complexity given in the statement of the theorem.
    
     It remains to compute $\val^*(u_0)$ knowing $\val^*(u_0, w)$ and $\ropti{v_0}{w}$ for every $w \in \Gamma^+(u_0)$. For this, we run \autoref{alg:OptimalStratBin} on arc $(z,u_0)$ with $z$ being a fictive vertex such that the only arc incident to $z$ is $(z,u_0)$, $r$ and $val$ are $\ropti{u_0}{w}$ and $\val^*(u_0, w)$ respectively for every $w \in \Gamma^+(u_0)$. Then the binary value returned by the algorithm is $\val^*(z, u_0)$. But in the $(z, u_0)$-subtree, the first step leads the particle to vertex $u_0$ and then never goes back to $z$: the run is then similar to the run starting at $u_0$ in the tree-like rotor game. Hence $\val^*(z, u_0) = \val^*(u_0)$.

\end{proof}

Some remarks are in order here. $(i)$ The algorithm used to compute $\val^*(u_0)$ provides the optimal value of $u_0$ as well as a subtree optimal strategy at $u_0$ for every decisional vertex. A tutorial example is given in \autoref{exmp : Example1PGame}. 
$(ii)$ In the case of a simple graph, one can compute $\val^*(u_0)$ for every vertex $u_0$ with the same time complexity as in \autoref{thm:opt_complex} as detailed in \autoref{sec:simple}.

\newpage
\begin{figure}[ht!]

    \begin{subfigure}[c]{\textwidth}
    \centering
    \begin{tikzpicture}[scale=0.9]

    \node[shape=circle,draw=black,thick] (G2) at (0,4) {$0$};
    \node[shape=circle,draw=black] (G) at (0,4) {$~~~~~$};
    \node[shape=circle,draw=black, very thick] (I) at (2,2) {$u_0$};
    \node[shape=rectangle,draw=black,thick] (E) at (2,4) {$u_4$};
    \node[shape=circle,draw=black] (J) at (4,4) {$u_5$};
    \node[shape=circle,draw=black, thick] (H2) at (6,4) {$1$};
    \node[shape=circle,draw=black] (H) at (6,4) {$~~~~~$};
    \node[shape=circle,draw=black] (K) at (6,2) {$~~~~~$};
    \node[shape=circle,draw=black, thick] (K2) at (6,2) {$0$};
    \node[shape=circle,draw=black,thick] (A2) at (0,0) {$0$};
    \node[shape=circle,draw=black] (A) at (0,0) {$~~~~~$};
    \node[shape=circle,draw=black] (B) at (2,0) {$u_1$};
    \node[shape=rectangle,draw=black,thick] (C) at (4,0) {$u_2$};
    \node[shape=circle,draw=black] (M) at (4,2) {$u_6$};
    \node[shape=circle,draw=black] (D) at (6,0) {$u_3$};
    \node[shape=circle,draw=black,thick] (L2) at (8,0) {$1$};
    \node[shape=circle,draw=black] (L) at (8,0) {$~~~~~$};
    
   \path [->, >=latex, bend left=30, very thick,dashed, red](B) edge (C);
   \path [->, >=latex, bend left=30](C) edge (B);
   \path [->, >=latex, bend left=30](D) edge (C);
   \path [->, >=latex, bend left=30](C) edge (D);
   \path [->, >=latex, bend left=30](B) edge (I);
   \path [->, >=latex, bend left=30](I) edge (B);
   \path[->,>=latex](B) edge (A);
   \path[->,>=latex, bend left=30, very thick,dashed, red](I) edge (E);
   \path[->,>=latex, bend left=30](E) edge (I);
   \path[->,>=latex, bend left=30](J) edge (E);
   \path[->,>=latex, bend left=30](E) edge (J);
   \path[->,>=latex, very thick,dashed, red](D) edge (L);
   \path[->,>=latex](E) edge (G);
   \path[->,>=latex](J) edge (H);
   \path[->,>=latex, very thick, red,dashed](J) edge (K);
   \path[->,>=latex, very thick,dashed, red,bend right=30](M) edge (C);

\draw [ thick,->,>=stealth, red](9,2) arc (0:330:0.4cm);
\end{tikzpicture}
\caption{Consider the one-player tree-like rotor game above with $u_0$ being the starting vertex, $V_{\Max}=\{u_2,u_4\}$ (depicted by squares), $V_r=\{u_0,u_1,u_3,u_5,u_6\}$, $\rho$ being the partial configuration on vertices of $V_r$ (depicted by the red arcs in dashes) and $\theta$ being the anticlockwise order on the outgoing arcs of each vertex (depicted by the cycling arrow on the right). The set $S_0$ contains all other vertices. Each of these sinks is represented by two circles with its value written inside.}

\end{subfigure}

\begin{subfigure}[c]{\textwidth}
    \centering
    \begin{tikzpicture}[scale=0.9]

    \node[shape=circle,draw=black,thick] (G2) at (0,4) {$0$};
    \node[shape=circle,draw=black] (G) at (0,4) {$~~~~~$};
    \node[shape=circle,draw=black, very thick] (I) at (2,2) {$u_0$};
    \node[shape=rectangle,draw=black,thick] (E) at (2,4) {$u_4$};
    \node[shape=circle,draw=black] (J) at (4,4) {$u_5$};
    \node[shape=circle,draw=black, thick] (H2) at (6,4) {$1$};
    \node[shape=circle,draw=black] (M) at (4,2) {$u_6$};
    \node[shape=circle,draw=black] (H) at (6,4) {$~~~~~$};
    \node[shape=circle,draw=black] (K) at (6,2) {$~~~~~$};
    \node[shape=circle,draw=black, thick] (K2) at (6,2) {$0$};
    \node[shape=circle,draw=black,thick] (A2) at (0,0) {$0$};
    \node[shape=circle,draw=black] (A) at (0,0) {$~~~~~$};
    \node[shape=circle,draw=black] (B) at (2,0) {$u_1$};
    \node[shape=rectangle,draw=black,thick] (C) at (4,0) {$u_2$};
    \node[shape=circle,draw=black] (D) at (6,0) {$u_3$};
    \node[shape=circle,draw=black,thick] (L2) at (8,0) {$1$};
    \node[shape=circle,draw=black] (L) at (8,0) {$~~~~~$};
    
   \path [->, >=latex, bend left=30, very thick,dashed, red](B) edge (C);
   \path [->, >=latex, bend left=30](C) edge (B);
   \path [->, >=latex, bend left=30](D) edge (C);
   \path [->, >=latex, bend left=30](C) edge (D);
   \path [->, >=latex, bend left=30](B) edge (I);
   \path [->, >=latex, bend left=30](I) edge (B);
   \path[->,>=latex](B) edge (A);
   \path[->,>=latex, bend left=30, very thick,dashed, red](I) edge (E);
   \path[->,>=latex, bend left=30](E) edge (I);
   \path[->,>=latex, bend left=30](J) edge (E);
   \path[->,>=latex, bend left=30](E) edge (J);
   \path[->,>=latex, very thick,dashed, red](D) edge (L);
   \path[->,>=latex](E) edge (G);
   \path[->,>=latex](J) edge (H);
   \path[->,>=latex, very thick,dashed, red](J) edge (K);
   \path[->,>=latex, very thick,dashed, red,bend right=30](M) edge (C);
   
    \node at (1.1,-0.3) {\scriptsize{(0,1)}};
    \node at (5,0.6) {\scriptsize{(1,1)}};
    \node at (6.9,-0.3) {\scriptsize{(1,1)}};
    \node at (4.8,4.3) {\scriptsize{(1,1)}};
    \node at (4.5,3) {\scriptsize{(0,1)}};
    \node at (1.1,4.3) {\scriptsize{(0,1)}};
    \node at (3,4.5) {\scriptsize{(0,1)}};

\draw [ thick,->,>=stealth, red](9,2) arc (0:330:0.4cm);
\end{tikzpicture}
\caption{We run the first 8 steps that does not involve a positional vertex of our algorithm recursively from the leaves and write the couple $(\val^*(u,v),\ropti{u}{v})$ next to each arc $(u,v)$ directed from $u_0$ towards the leaves. }

\end{subfigure}

\begin{subfigure}[c]{\textwidth}
    \centering
    \begin{tikzpicture}[scale=0.9]

    \node[shape=circle,draw=black,thick] (G2) at (0,4) {$0$};
    \node[shape=circle,draw=black] (G) at (0,4) {$~~~~~$};
    \node[shape=circle,draw=black, very thick] (I) at (2,2) {$u_0$};
    \node[shape=rectangle,draw=black,thick] (E) at (2,4) {$u_4$};
    \node[shape=circle,draw=black] (J) at (4,4) {$u_5$};
    \node[shape=circle,draw=black, thick] (H2) at (6,4) {$1$};
    \node[shape=circle,draw=black] (M) at (4,2) {$u_6$};
    \node[shape=circle,draw=black] (H) at (6,4) {$~~~~~$};
    \node[shape=circle,draw=black] (K) at (6,2) {$~~~~~$};
    \node[shape=circle,draw=black, thick] (K2) at (6,2) {$0$};
    \node[shape=circle,draw=black,thick] (A2) at (0,0) {$0$};
    \node[shape=circle,draw=black] (A) at (0,0) {$~~~~~$};
    \node[shape=circle,draw=black] (B) at (2,0) {$u_1$};
    \node[shape=rectangle,draw=black,thick] (C) at (4,0) {$u_2$};
    \node[shape=circle,draw=black] (D) at (6,0) {$u_3$};
    \node[shape=circle,draw=black,thick] (L2) at (8,0) {$1$};
    \node[shape=circle,draw=black] (L) at (8,0) {$~~~~~$};
    
   \path [->, >=latex, bend left=30, very thick,dashed, red](B) edge (C);
   \path [->, >=latex, bend left=30](C) edge (B);
   \path [->, >=latex, bend left=30](D) edge (C);
   \path [->, >=latex, bend left=30, very thick, blue](C) edge (D);
   \path [->, >=latex, bend left=30](B) edge (I);
   \path [->, >=latex, bend left=30](I) edge (B);
   \path[->,>=latex](B) edge (A);
   \path[->,>=latex, bend left=30, very thick,dashed, red](I) edge (E);
   \path[->,>=latex, bend left=30,very thick, blue](E) edge (I);
   \path[->,>=latex, bend left=30](J) edge (E);
   \path[->,>=latex, bend left=30](E) edge (J);
   \path[->,>=latex, very thick,dashed, red](D) edge (L);
   \path[->,>=latex](E) edge (G);
   \path[->,>=latex](J) edge (H);
   \path[->,>=latex, very thick,dashed, red](J) edge (K);
   \path[->,>=latex, very thick,dashed, red, bend right=30](M) edge (C);
   
    \node at (1.1,-0.3) {\scriptsize{(0,1)}};
    \node at (5,0.6) {\scriptsize{(1,1)}};
    \node at (6.9,-0.3) {\scriptsize{(1,1)}};
    \node at (4.8,4.3) {\scriptsize{(1,1)}};
    \node at (4.5,3) {\scriptsize{(0,1)}};
    \node at (1.1,4.3) {\scriptsize{(0,1)}};
    \node at (1.2,3) {\scriptsize{(0,2)}};
    \node at (3,4.5) {\scriptsize{(0,1)}};
    \node at (3,0.6) {\scriptsize{(1,1)}};
    \node at (2.8,1.1) {\scriptsize{(1,1)}};

\draw [ thick,->,>=stealth, red](9,2) arc (0:330:0.4cm);
\end{tikzpicture}
\caption{We proceed three more steps of our algorithm and compute $\sigma^*(u_2)$ and $\sigma^*(u_6)$ while doing so. The strategy $\sigma^*$ is depicted by blue arcs. From here, we deduce that $\val^*(u_0)=1$ thanks to \autoref{thm:opt_complex}.}

\end{subfigure}
\caption{Computation of $\val^*$ and $r^*$ using \autoref{thm:opt_complex}.}
\label{exmp : Example1PGame}
\end{figure}

\newpage

\RestyleAlgo{ruled}
\SetKwInOut{Input}{input}\SetKwInOut{Output}{output}

\begin{algorithm}[!ht]
\caption{Optimal Strategy Computation}\label{alg:OptimalStratBin}
\DontPrintSemicolon
\Input{$(u,v)$ is an arc of $\arcs$; $r : \Gamma^+(v) \setminus \{u\} \rightarrow \mathbb Z \cup {+ \infty}$ contains the optimal return flows of arcs $(v,w) \in \hat{\cal A}$; $val: \Gamma^+(v) \setminus \{u\} \rightarrow \{0,1\}$ contains the optimal values of these arcs.}
\Output{an element of ${\cal A}^+(v)$ (optimal strategy $\sigma^*(v)$); a positive integer ($\ropti{u}{v}$) and a binary integer ($\val^*(u,v)$))}
\BlankLine

For $w \in \Gamma^+(v)$, let $Q(w)$ be the quotient of the euclidean division of $r(w)$ by $|h^{-1}(w)|$; let $q_{\text{min}}$ be the minimum value of $Q(w)$ for all $w$; let $R(w)$ be $r(w)-(q_{\text{min}}*|h^{-1}(w)|)$; let $a_s=a_e=a_0$ with $a_0$ being any arc of ${\cal A}^+(v)$,
let $f_0,f_1, f_u, value$ be integers initialized to $0$. Let $b_0,b_1$ be two arcs initialized to $a_0$. Let \emph{last} be a boolean with initial value False.
\BlankLine

\If{$r(w) = +\infty,  \forall w \in \Gamma^+(v) \backslash \{u\}$}{
    \Return{$b_0, +\infty ,0$}
}
\BlankLine

\Repeat{$a_s = a_0$ and last=True}{
  \eIf{$h(a_e) = u$}{
    $f_u \mathrel{+}=1$ \;
    $a_e \gets \theta_v(a_e)$\;
  }
  {
    $R(h(a_e)) \mathrel{-}=1$\;
      \eIf{$R(h(a_e))= 0$}{
        \eIf{$val(a_e)= 0$}{
        \Comment{Keep track of the maximal return with value $0$}
            \If{$f_0 < f_u$}{
            $f_0 \gets f_u$\;
            $b_0 \gets a_s$\;
            }
        }{
        \If{$f_1 > f_u$}{\Comment{Keep track of the minimal with value $1$ and recall the fact that there exist a strategy with value $1$}
            $f_1 \gets  f_u$\;
            $b_1 \gets a_s$\;
        }
        $value \gets 1$\;
        }
        \eIf{$h(a_s) = u$}{
            $f_u \mathrel{-}=1$ 
        }
        {
            $R(h(a_s)) \mathrel{+}= 1$\;
        }
    $a_s \gets \theta_v(a_s)$\;
    \If{$\theta_v(a_s)=a_0$}{last $\gets$ True}
    }
    {
        $a_e \gets \theta_v(a_e)$\;
    }
  }
}

\eIf{$value=1$}{
\Return{$b_1, f_1 + q_{\text{min}}*|h^{-1}(u)|+1, 1$}
}
{
\Return{$b_0,f_0 + q_{\text{min}}*|h^{-1}(u)|+1 ,0$}
}
\end{algorithm}

\newpage

\subsection{One-player Integer Rotor Game}
\label{subsec:oneP-integ}

We now turn to the case where $\val(s)$ is no longer restricted to be a binary value. The main difference with the binary case is that there may not exist a subtree optimal strategy as illustrated in  \autoref{nonBinarStrat}. In this example the value $1$ is an intermediate sink value (neither maximal not minimal), hence it cannot be decided with the only knowledge of the  subtree optimal strategy on the  $(v_0,v)$-subtree whether the return flow should be minimized in order to try reaching this sink or maximized if a sink with higher value can be reached in the rest of the graph. 
The knowing of $\val^*(u,v)$ and $r^*(u,v)$ that was enough in the binary case, is not sufficient anymore for computing the optimal value recursively. In the case of simple graphs, we show in \autoref{subsubsec:1PSimpleInt} that we can add information on the subtrees in order to compute the optimal value in linear time complexity, but this technique does not extend to multigraphs.

\begin{comment}
Indeed, in the binary case we had two types of sinks those with value $0$ and with value $1$ which are respectively indistinguishable. Hence, reaching a sink of value $1$ ensures that there is no sink of greater value in the graph. But, for the integer case, reaching a sink with value $x$ does not guarantee that there is no sink with greater value $x'$ elsewhere in the graph. Therefore, the algorithm that worked for binary values does not apply here. 
\end{comment}

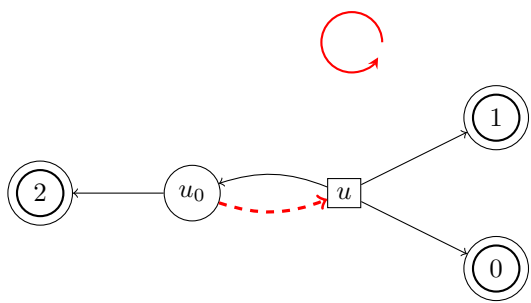
\begin{figure}[!ht]
    \centering
    \begin{tikzpicture}
         
    \node[shape=circle,draw=black] (A) at (0,0) {$u_0$};
    \node[shape=rectangle,draw=black] (B) at (2,0) {$u$};
    \node[shape=circle,draw=black,thick] (C2) at (4,-1) {$0$};
    \node[shape=circle,draw=black] (C) at (4,-1) {$~~~~~$};
    \node[shape=circle,draw=black,thick] (C'2) at (4,1) {$1$};
    \node[shape=circle,draw=black] (C') at (4,1) {$~~~~~$};
    
    \node[shape=circle,draw=black,thick] (E2) at (-2,0) {$2$};
    \node[shape=circle,draw=black] (E) at (-2,0) {$~~~~~$};
    
    \path [->,bend right=20](B) edge (A);
    \path [->,bend right=20, red, dashed, very thick](A) edge (B);
    \path [->](A) edge (E);
    \path [->](B) edge (C);
    \path [->](B) edge (C');

   \draw [ thick,->,>=stealth, red](2.5,2) arc (0:330:0.4cm);

\end{tikzpicture}
\caption{Example of a tree-like rotor game that does not admit a subtree optimal strategy. Here, $u \in V_{\Max}$ and $u_0 \in V_r$. All other vertices belong to $S_0$ and their value is written inside them.
The initial configuration of $u_0$ is the red arc in dashes. In the $(u_0,u)$-subtree, the only optimal strategy is to direct $u$ towards the sink of value $1$ which gives value $1$ with return flow $1$.  In the entire game with starting vertex $u_0$, the value of this strategy is $1$. But the optimal value is $2$ which is obtained by directing $u$  towards the sink of value $2$. }
\label{nonBinarStrat}
\end{figure}

Despite that, \autoref{thm:opt_complex} can be used as a basis of a bisection (dichotomy) method for computing the optimal value. For this, consider the decision problem of determining whether a sink of value at least $x$ can be reached. We can solve it by introducing a binary game obtained by replacing all values that are greater or equal to $x$ by one and the others by zero. It should be clear that the value of the binary game is one if and only if a sink of value at least $x$ can be reached in the initial game. All in all, the non-binary game can be solved in $O(\log(|S_0|))$ such iterations. Knowing the optimal value, say $v^*$, an optimal strategy can be computed by solving the binary associated game where threshold $x$ is chosen equal to $v^*$.

\subsection{One-player Rotor Game: Other Set of Strategies}

Here we briefly discuss some variants of this game, where the player can choose from a different set of strategies than just the initial configurations on the vertices of $V_{\Max}$.
\begin{enumerate}
    \item Let us consider the set of strategies where the player can freely decide the rotor order on vertices he controls and the starting configuration on it. Consider an arc $(u,v) \in \arcs$ and assume that $\ropti{v}{w}$ and \emph{$\val^*(v,w)$} are known for every $w \in \Gamma^+(v) \setminus{u}$. One can compute a subtree optimal strategy on the $(u,v)$-subtree in the same way than for the previous binary case but where the return flows are either $\rf{u}{v}{\sigma}=q_{min}*|{\cal A}_{(u,v)}|$ if $\val^*(u,v)=1$ (all occurrences of an arc with head $u$ are placed at the end of the rotor order) or $\rf{u}{v}{\sigma}=(q_{min}+1)*|{\cal A}_{(u,v)}|$ if $\val^*(u,v)=0$ (all occurrences of an arc with head $u$ are placed at the beginning of the rotor order). This also simplifies the integer case consequently.
    \item Let us consider the infinite set of strategies where the player can choose at each step of the rotor walk the orientation of the vertices he controls (as in \cite{dohrau2017arrival}). For a vertex $v \in V_{\Max}$, the player can choose to put the return flow on any outgoing arc of $v$ to either $1$ or $+\infty$. Consider an arc $(u,v) \in \arcs$ and assume that $\ropti{v}{w}$ and \emph{$\val^*(v,w)$} are known for every $w \in \Gamma^+(v) \setminus{u}$. A subtree optimal strategy is easily computed by choosing the strategy to always go towards a vertex $w$ if $\val^*(v,w)=1$ and if not to choose the strategy that maximizes the value of  $\rf{u}{v}{\sigma}$. In particular, if there exists an arc $(v,u) \in \arcs$ the strategy on $v$ would be to always go towards $(v,u)$. Once again, the integer case is simplified consequently.
\end{enumerate}

\section{Two-player Rotor Game}
\label{sec:ARRIVAL2P}

We now consider a two-player, zero-sum version of the zero-player game, where players control distinct subsets of vertices of $V_0$, one trying to maximize the value of the sink that has been reached whereas the other one tries to minimize it. A similar game (but where players freely decide the orientation of their vertices at each time step) has been studied in~\cite{fearnley2017reachability} and shown to be P-SPACE complete.

More formally, a two-player rotor game is given by  \[G=(V_r, V_{\Max}, V_{\Min}, S_0, {\cal A}, h, t, \theta, \val, \rho)\]
where $V_r, V_{\Max}, V_{\Min}$ and $S_0$ are disjoint sets of vertices, such that:
\begin{itemize}
    \item for all partial configurations $\tau$ on $V_{\Min}$, 
    \[G(\cdot,\tau) = (V_r \cup V_{\Min}, V_{\Max}, S_0, {\cal A}, h, t, \theta, \val, (\rho, \tau))\]
    is a one-player rotor game;
    \item for all partial configurations $\sigma$ on $V_{\Max}$,  \[G(\sigma,\cdot) = (V_r \cup V_{\Max}, V_{\Min}, S_0, {\cal A}, h, t, \theta, \val, (\rho, \sigma))\] is  a one-player rotor game.
\end{itemize}

If all the one-player games are tree-like (in other words, if the underlying graph is tree-like) then the two-player game is also said to be tree-like. 

A strategy for player $\Max$ (respectively player $\Min$) is a partial rotor configuration on $V_{\Max}$ (respectively on $V_{\Min}$). We denote by $\Sigma_{\Max}$ and $\Sigma_{\Min}$ the sets of  strategies for these players. When $\tau$ is fixed, $\Max$ tries to maximize the value of the final sink in $G(\cdot,\tau)$; whereas when $\sigma$ is fixed, $\Min$ tries to minimize it in $G(\sigma,\cdot)$.

The {\it value of the game} for strategies $\sigma, \tau$  and starting vertex $u_0$ is denoted by $\val_{\sigma, \tau}(u_0)$ and is the value $\val(s)$ where $s$ is the sink reached by a maximal rotor walk from the rotor particle configuration $((\rho, \sigma, \tau),u_0)$ if any, or $0$ otherwise. As we did for one-player, we assume in the following that all rotor games are tree-like and stopping.

When $u_0$ and $\rho$ are fixed, this defines a zero-sum game where $\Max$ and $\Min$ try respectively to maximize and minimize the value of the game by choosing an appropriate strategy, respectively in $\Sigma_{\Max}$ and $\Sigma_{\Min}$.
Usually, such a zero-sum game does not always have an equilibrium in pure strategies and so-called mixed (i.e. stochastic) strategies are required; this is the case in the example of \autoref{NoEquilibrium} where the given graph is not tree-like. However, in the case of tree-like multigraphs we prove the following theorem.

\begin{theorem}[Existence of pure strategy Equilibrium]
\label{thm:pure_eq}
Let $G$ be a tree-like two-player rotor game  together with a starting vertex $u_0$. Then there are an integer value $\val^*$ and two strategies $\sigma^*, \tau^*$ such that
\begin{enumerate}[(i)]
    \item $\forall \tau \in \Sigma_{\Min}, \ \val_{\sigma^*, \tau}(u_0) \geq \val^*$, \emph{i.e.} $\tau^*$ is optimal in the one-player game $G(\sigma^*, \cdot)$,
    \item  $\forall \sigma \in \Sigma_{\Max}, \ \val_{\sigma, \tau^*}(u_0) \leq \val^*$, \emph{i.e.} $\sigma^*$ is optimal in the one-player game $G(\cdot, \tau^*)$.
\end{enumerate}
We call $\val^*$ the value of the game and the pair $(\sigma^*, \tau^*)$ is a pure strategy equilibrium.

Furthermore $\val^*$ can be computed in the same time complexity as the computation of $D(\rho)$ in the zero-player game (see \autoref{CDA}).
\end{theorem}

This theorem is proved by following the same scheme as for the one player game: first, we consider the binary case, and then the general case follows by dichotomy. A constructive proof is given in both cases.

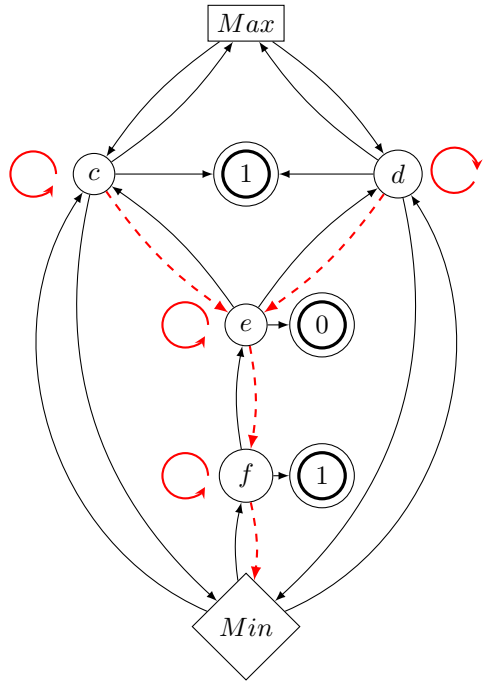
\begin{figure}[!ht]
    \centering
    \begin{tikzpicture}
         
    \node[shape=rectangle,draw=black] (I) at (2,2) {$Max$};
    \node[shape=circle,draw=black] (E) at (2,-2) {$e$};
    \node[shape=circle,draw=black] (C) at (0,0) {$c$};
    \node[shape=circle,draw=black,very thick] (B2) at (2,0) {$1$};
    \node[shape=circle,draw=black] (B) at (2,0) {$~~~~~$};
    \node[shape=circle,draw=black] (D) at (4,0) {$d$};
    \node[shape=circle,draw=black] (F) at (2,-4) {$f$};
    \node[shape=diamond,draw=black] (J) at (2,-6) {$Min$};
    \node[shape=circle,draw=black,very thick] (A2) at (3,-2) {$0$};
    \node[shape=circle,draw=black] (A) at (3,-2) {$~~~~~$};
    \node[shape=circle,draw=black,very thick] (H2) at (3,-4) {$1$};~
    \node[shape=circle,draw=black] (H) at (3,-4) {$~~~~~$};
    
   \path [->, bend right=50, >=latex](J) edge (D);
    \path [->, bend left=30, >=latex](D) edge (J);
     \path [->, bend left=50, >=latex](J) edge (C);
    \path [->, bend right=30, >=latex](C) edge (J);
   \path [->, bend right=10, >=latex] (I) edge (C);
   \path [->, bend right=10, >=latex] (C) edge (I);
   \path [->, bend left=10, >=latex] (I) edge (D);
   \path [->, bend left=10, >=latex] (D) edge (I);
   \path [->,thick,red, dashed, bend left=10, >=latex] (E) edge (F);
   \path [->, bend left=10, >=latex] (F) edge (E);
   \path [->,thick,red, dashed, bend left=10, >=latex] (F) edge (J);
   \path [->, bend left=10, >=latex] (J) edge (F);
   \path [->, >=latex](C) edge (B);
   \path [->, >=latex](D) edge (B);
   \path [->, >=latex](E) edge (A);
   \path [->, >=latex](F) edge (H);
   \path [->,thick,red, dashed, bend right=10, >=latex] (C) edge (E);
   \path [->, bend right=10, >=latex] (E) edge (C);
   \path [->,thick,red, dashed, bend left=10, >=latex] (D) edge (E);
   \path [->, bend left=10, >=latex] (E) edge (D);

   \draw [ thick,->,>=stealth, red](-0.5,0) arc (0:330:0.3cm);
   \draw [ thick,->,>=stealth, red](1.5,-2) arc (0:330:0.3cm);
   \draw [ thick,->,>=stealth, red](1.5,-4) arc (0:330:0.3cm);
   \draw [ thick,->,>=stealth, red](5,-0.1) arc (330:0:0.3cm);
   
\end{tikzpicture}
\caption{This example is a simple undirected graph where each edge is replaced by two arcs. We have $V_{\Max}=\{Max\}$, $V_{\Min}=\{Min\}$ and $V_r=\{c,d,e,f\}$ and $S_0$ is the rest of the vertices with their value written inside them. The particle starts on the vertex $Max$. 
In this game, the only optimal strategy for $\Max$ when the strategy for $\Min$ is the arc $(Min, x)$ with $x \in \{c, d \}$ is $(Max, x)$. 
On the other hand the only optimal strategy for $\Min$ when the strategy for $\Max$ is the arc $(Max, c)$ (resp. $(Max, d)$) is $(Min, d)$ (resp. $(Min, c)$). The situation is like the classical matching pennies game (see \cite{osborne2004introduction} for precise definition) where one player tries to match the strategy of the opponent whereas the other player has the opposite objective. It is known that such game does not admit a Nash equilibrium in pure strategies. }
\label{NoEquilibrium}
\end{figure}

\subsection{Two-player Binary Rotor Game}
\label{sec:2P-B}
In this subsection, we restrict the game to the case where values of sinks are binary numbers i.e. $\val(s) \in \{0,1\}$ for all $s \in S_0$. 

\begin{definition}[Subtree equilibrium]
A pair of strategies $\sigma^*, \tau^*$ is a \emph{subtree equilibrium} at $u_0$ if 
$\sigma^*$ (resp. $\tau^*$) is subtree optimal at $u_0$ in the one-player game $G(\cdot,\tau^*)$ (resp. $G(\sigma^*,\cdot)$).
\end{definition}

\begin{lemma}
For any $(u,v) \in \arcs$ and $u_0 \in V$, there is a subtree equilibrium at $u_0$ for the two-player game on $T_{(u,v)}$.
\end{lemma}

\begin{proof}
This is showed by induction.

This is clearly true if $v\in S_0$.

Otherwise, let $(u,v) \in \arcs$. Assume that there is a subtree equilibrium $(\sigma^*_w, \tau^*_w)$ for every $(v,w)$-subtree with $w \in \Gamma^+(v)\backslash \{u\}$. Consider the pairs of strategies $\sigma_v, \tau_v$ defined on $T_{(u,v)}$ such that $\sigma_v(z) = \sigma^*_w(z)$ if $z\neq v$ and  $z \in T_{(v,w)}$ and  $\tau_v(z) = \tau^*_w(z)$ if $z\neq v$ and  $z \in T_{(v,w)}$.

Using \autoref{lem:propag_val}:
\begin{itemize}
    \item if $v \in V_r$ then $\sigma_v$ (resp. $\tau_v$) is subtree optimal in $T_{(u,v)}(\cdot,\tau_v)$ (resp. $T_{(u,v)}(\sigma_v,\cdot)$);
    \item if $v \in V_{\Max}$ then we choose $\sigma_v(v)$ as in \autoref{lem:propag_val} so that $\sigma_v$ is subtree optimal in $T_{u,v}(\cdot, \tau_v)$; on the other side $\tau_v$ remains subtree optimal in $T_{u,v}(\sigma_v,\cdot)$;
    \item we proceed similarly if $v \in V_{\Min}$.
\end{itemize} 
\end{proof}

This proves the existence of pure strategy Equilibrium as stated in \autoref{thm:pure_eq}.  Using a minimax-like algorithm, we get the time complexity claimed in the theorem.

\subsection{Two-player Integer Rotor Game}

By analogy with the one-player case, we can search the value of the game $\val^*$ in a dichotomic way. Recall that it is defined as the maximal value that player $\Max$ can guarantee against any strategy of $\Min$. 
To know if $\Max$ can guarantee at least value $x$, it suffices to solve the binary game obtained by replacing all values of sinks greater or equal than $x$ by $1$ and the other by  $0$. Then, the value of the binary game is $1$ if and only if the value of the initial game is at least $x$.  By iterating this, $\val^*$ can be determined by solving $O(\log(|S_0|))$ binary games. 

Once $\val^*$ is known, we found no simple process to deduce an equilibrium from the value of the game: this remains an open problem.

\section{Simple Graph}\label{sec:simple}
\label{sec:ARRIVALSimple}

In this section we consider simple graphs i.e. graphs such that for each vertex $u \in V$ there is at most on arc with head $v$ when $v \in \arcs^+(u)$.
We show specific results for this kind of graphs: firstly we give a formula allowing to propagate the return flow without using the routine and then we show that this helps to efficiently solve the decisional framework when  values are not restricted to be binary.

We introduce some notations that will be useful in the rest of this section. Consider two arcs $b,c$ of $\arcs^+(u)$, in order to measure which one is reached first by successive rotor order operations starting on a particular arc $a \in \arcs^+(u)$ we define the \emph{distance in a rotor orbit}.

\begin{definition}[Distance in a rotor orbit]
\label{def:distance}
Given a vertex $u$ and $a,b \in \arcs^+(u)$, we denote by $d_u(a,b)$ the smallest $i \geq 0$ such that $\theta_u^i(a)=b$.
\end{definition}

The following operator basically checks which arc between $b$ and $c$ is encountered first while listing the orbit of $\theta_u$ starting from $a$. 

\begin{definition}
\label{projection}
 Let $a,b,c$ be arcs of $\arcs^+(u)$. We define the operator $B_u$ such that $B_u(a,b,c)=1$ if $d_u(a,b) \leq d_u(a,c)$ and $B_u(a,b,c)=0$ if not.
\end{definition}

\subsection{Zero-player Game}
\label{subsec : ReturnSimple}

In a simple graph, by definition $|\Gamma^+(u)|$ is equal to $|\arcs^+(u)|$. So each time that the rotor makes one full turn at vertex $u$, the particle travels exactly once between $u$ and $v$ for each $v \in \Gamma^+(u)$. Consequently, we have the following result.

\begin{lemma}
\label{lem:nbrVisits}
If $D(\rho)(u) = (u,v)$ then all arcs $(u,w)$ of $\arcs^+(u)$ are visited exactly 
\[F_\rho(u,w)=\rf{u}{v}{\rho} \mathrel{-}B_u(\rho(u),(u,v),(u,w))\] times during a maximal rotor walk starting from $(\rho,u)$.
\end{lemma}

This  allows to compute $D(\rho)(u)$ from the return flows of the outgoing arcs of $u$, as stated in next lemma.

\begin{lemma}
\label{lem:ret2dest}
Among all arcs of $\arcs^+(u)$ which have a minimal return flow,
the arc $D(\rho)(u)$ is the first one with respect to order $\theta_u$ starting at $\rho(u)$. The flow on all arcs $(u,v) \in \arcs$ can be computed in time $O(|\Gamma^+(u)|)$.
\end{lemma}

\begin{proof}
Assume that $D(\rho)(u) = (u,v)$.
From \autoref{lem:nbrVisits},  all arcs $(u,w)$ are visited exactly $F_{\rho}(u,w)=\rf{u}{v}{\rho} \mathrel{-}B_u(\rho(u),(u,v),(u,w))$ times during the maximal rotor walk starting from $(\rho, u)$.
Note that, for all $w \in \Gamma^+(u)$, $B_u(\rho(u),(u,v),(u,w))$ can be computed in time $O(|\Gamma^+(u)|)$ by applying $\theta_u$ successively starting from $\rho(u)$.

Then, for all $w \in \Gamma^+(u)\setminus{v}$, we have:
$$\rf{u}{w}{\rho} > F_\rho(u, w) = \rf{u}{v}{\rho} \mathrel{-}B_u(\rho(u),(u,v),(u,w)) \geq \rf{u}{v}{\rho} - 1$$
where the first inequality comes from \autoref{lem:flowReturnFlow}. This gives $\rf{u}{v}{\rho} < 1 + \rf{u}{w}{\rho}$, or, equivalently, $\rf{u}{v}{\rho} \leq \rf{u}{w}{\rho}$. Hence $\rf{u}{v}{\rho}$ is minimal.

On the other side, if $\rf{u}{w}{\rho} = \rf{u}{v}{\rho}$ then $B_u(\rho(u),(u,v),(u,w)) = 1$, which means that $(u,v)$ is the first arc with respect to order $\theta_u$, starting at $\rho(u)$, among all arcs of $\arcs^+(u)$ with minimal return flow. 

In that process, we also show the second part of the result i.e. the flow on arc $(u,w)$ can be computed in time $O(|\Gamma^+(u)|)$.
\end{proof}

Let $(u,v) \in \arcs$ with $u \in V_0$. Consider the $(u,v)$-subtree and let $(v, w)$ be the orientation of $v$ in the Destination Forest in that subtree. From \autoref{lem:nbrVisits} we have $\rf{u}{v}{\rho} = \rf{v}{w}{\rho} + B_v(\rho(v),(v,u),(v,w))$. Combining this with \autoref{lem:ret2dest}, we obtain the following recursive computation of $\rf{u}{v}{\rho}$.

\begin{equation}
    \rf{u}{v}{\rho}= \min_{w \in \Gamma^+(v)\setminus\{u\}}\rf{v}{w}{\rho} + B_v(\rho(v),(v,u),(v,w))
    \label{form:returnSimple}
\end{equation}

We then have a result similar to that of  \autoref{lem: Retropropag}.

\begin{lemma}[Simple Retropropagation for the Zero-player Game]
\label{lem:SimpleRetroPropag}
 Given  $v \in V$, assuming the return flows of all arcs $(v, w)$ for $w \in \Gamma^+(v)$ are known, one can compute the return flow of all arcs $(z,v)$ with $z \in \Gamma^-(v)$ by applying Equation~\eqref{form:returnSimple} at most twice in time $O(\min(| \Gamma^+(v)|, | \Gamma^-(v)|))$.
\end{lemma}

\begin{proof}
For every vertex $z\in \Gamma^-(v) \setminus \Gamma^+(v)$ we set $\rf{z}{v}{\rho}=1$ (or $\rf{z}{v}{\rho}=0$ if $z \in S_0)$.

As the graph is stopping, there is at least one vertex $w \in \Gamma^+(v)$ such that the return flow of $(v,w)$ is finite. Let $(v,w_0)=D_\rho(v)$ be the first arc with respect to order $\theta_v$ starting from $\rho(v)$ such that $\rf{v}{w_0}{\rho}=\min_{w \in \Gamma^+(v)}\rf{v}{w}{\rho}$. Then, for any $(w,v)$-subtree with $w \in \Gamma^-(v)$ and  $w \neq w_0$, the last outgoing arc of $v$ when routing a particle from $(\rho,v)$ is $(v,w_0)$, hence by Equation~\eqref{form:returnSimple}, we have $\rf{w}{v}{\rho}=\rf{v}{w_0}{\rho}+B_v(\rho(v),(v,w),(v,w_0))$. Note that all values $B_v(\rho(v),(v,w),(v,w_0))$ can be computed in  time $O(| \Gamma^+(v)|)$ by iterating $\theta_v$ from $\rho(v)$.
It remains to compute $\rf{w_0}{v}{\rho}$ which can be done by applying Equation~\eqref{form:returnSimple} once more. Potentially, the return flows of all arcs $(v,w)$ with $w \in \Gamma^+(v)\setminus\{w_0\}$ might be infinite. In that case, we have $\rf{w_0}{v}{\rho}=+\infty$.
\end{proof}

From, here, we have all the tools to construct our recursive algorithm.

\begin{theorem}[Complete Destination Algorithm for simple graphs]
\label{CDA-simple}
The configuration $D(\rho)$ can be computed in time complexity $O(|{\cal V}|)$.
\end{theorem}

\begin{proof}
We use the same algorithm as for the multigraph case (see \autoref{CDA}), except that we replace routine IRR used for propagating the return flow by Equation~\eqref{form:returnSimple}, which is more efficient since it does not need a division. And \autoref{lem: Retropropag} is replaced by \autoref{lem:SimpleRetroPropag}. Since the graph is simple, we have $|\arcs| = |{\cal A}|= O(|{\cal V}|)$ hence the result.
\end{proof}

\subsection{One-player Simple Tree-like Rotor Game}
\label{subsec:onePSimple}

We now consider one-player rotor games when the graph is simple. For a simple tree-like rotor game, with binary values, we show that we can compute the optimal value of every starting vertex in the same time than in \autoref{CDA-simple}. Then, contrary to the case presented \autoref{subsec:oneP-integ}, for a simple game with integer values, we can achieve linear complexity to compute the value of the game.

First, \autoref{lem:simpleReturnMin} is a direct consequence of \autoref{lem:ret2dest}. It characterizes a subset of arcs likely to be the last outgoing arc of $v$.

\begin{lemma}
\label{lem:simpleReturnMin}
Consider a one-player simple tree-like rotor game $G$, a vertex $u \in V_0$ and a strategy $\sigma$. Then, we have $h(D(\rho,\sigma)(u)) \in \text{argmin}_{v\in\Gamma^+(u)}\rf{u}{v}{\sigma}$. 
\end{lemma}

This is true in particular for an optimal strategy of the one-player simple tree-like rotor game. Thus, we can propagate the values and return flows more easily than for the one-player tree-like rotor  game.

\subsubsection{Binary Values}

Consider a one-player simple tree-like rotor game with binary values and a strategy $\sigma$ for the player.

\begin{lemma}[Optimal Values and  Return Flows]
\label{lem:OptimalValRetSimple}
Let $(u,v) \in \arcs$. Then:
\begin{enumerate}
    \item If $v \in V_r$, suppose that we know all optimal values $\val^*(v,w)$ and optimal return flows $\ropti{v}{w}$ of all arcs $(v,w)$ with $w \neq u$. From this we can compute $\val^*(u,v)$ and $\ropti{u}{v}$ by:
    \[\val^*(u,v)=\val^*(v,w_f) \text{ and } \ropti{u}{v}=\ropti{v}{w_f}+B_v(\rho(v),(v,u),(v,w_f))\]
    where $w_f = h(D(\rho,\sigma)(v))$, i.e. $(v,w_f)$ is the last outgoing arc of $v$.
    \item If $v \in V_{\Max}$ and if there is at least one arc $(v,w_1)$ with $w_1 \in \Gamma^+(v) \setminus \{u\}$ such that $\val^*(v,w_1)=1$ and $\ropti{v}{w_1}=\min_{w \in \Gamma^+(v)\setminus\{u\}} \ropti{v}{w}$ then, $$\val^*(u,v)=1 \text{ and } \ropti{u}{v}=\left\{
 \begin{array}{ll}
    \ropti{v}{w_1} \text{ if } (v,u) \in \arcs;\\
    1 \text{ if not}.
    \end{array}
    \right.
    $$
    
    \item If $v \in V_{\Max}$ and if there is no arc $(v,w_1)$ as in case~2, we have $$\val^*(u,v)=0 \text{ and } \ropti{u}{v}=\left\{
 \begin{array}{ll}
    1+ \min_{w \in \Gamma^+(v)\setminus\{u\}} \ropti{v}{w} \text{ if } (v,u) \in \arcs;\\
    1 \text{ if not}.
    \end{array}
    \right.
    $$
\end{enumerate}
\end{lemma}

Note that $\ropti{u}{v}$ can be infinite.

\begin{proof}

\begin{enumerate}
    \item If $v \in V_r$, we use Equation~\eqref{form:returnSimple}.
    \item Let $(v,w_1)$ be an arc with $w_1 \in \Gamma^+(v) \setminus \{u\}$ such that $\val^*(v,w_1)=1$ and $\ropti{v}{w_1}=\min_{w \in \Gamma^+(v)\setminus\{u\}} \ropti{v}{w}$. \autoref{lem:nbrVisits} implies that there is a strategy $\sigma$ with $D(\rho,\sigma)=(v,w_1)$, hence  $\val^*(u,v)=1$. Furthermore, given any strategy $\sigma$ such that $\val_\sigma(u,v)=1$, and for all $w$ such that $\val^*(v,w)=1$ we have $\rf{v}{w}{\sigma} \geq \ropti{v}{w} \geq \ropti{v}{w_1}$, hence $\ropti{u}{v}=\ropti{v}{w_1}$.
    \item Otherwise, all arcs with minimal return flows have optimal value $0$. Obviously, this implies $\val^*(u,v)=0$. Given any strategy $\sigma$, for all $w$ such that $\val^*(v,w)=0$, we have $\ropti{v}{w} \geq \rf{v}{w}{\sigma}$ hence $\min_{w \in \Gamma^+(v)\setminus\{u\}} \ropti{v}{w} \geq \min_{w \in \Gamma^+(v)\setminus\{u\}} \rf{v}{w}{\sigma}$. Furthermore, if $(v,u) \in \arcs$, there is a strategy (namely choosing $(v,u)$ as the starting configuration of $v$) such that $(v,u)$ is visited $\min_{w \in \Gamma^+(v)\setminus\{u\}} \ropti{v}{w}+1$ times. Hence, we have $\ropti{u}{v}=1 +\min_{w \in \Gamma^+(v)\setminus\{u\}} \ropti{v}{w}$.
\end{enumerate}
\end{proof}

By analogy with \autoref{lem:propag_val}, we can construct a recursive linear algorithm to compute the optimal value of $u_0$. But, contrary to the one-player tree-like rotor game on multigraphs, as the graph is simple, we can also compute in linear time the return flows and values of all arcs of $\arcs$ that are not directed from $u_0$ towards the leaves. The reason we should do this is that, if we know all return flows of all arcs, in the end we can compute optimal values and strategies for all starting vertices simultaneously.
Indeed, by using conjointly \autoref{lem:SimpleRetroPropag} and \autoref{lem:OptimalValRetSimple} this can be done in at most $2|{\cal A}^+(u)|$ steps for a given vertex $u$. This result is properly stated in the following Lemma.

\begin{lemma}[Optimal Retropropagation]
\label{OptiRetroSimple}
Given a vertex $v \in V$ and assuming the return flows and values of all arcs $(v,w)$ for $w \in \Gamma^+(v)$ are known, one can compute all optimal return flows and all optimal values of arcs $(u,v)$ with $u \in \Gamma^-(v)$ by applying at most twice Equation~\eqref{form:returnSimple}.
\end{lemma}

\begin{proof}
If $v \in V_r$, this is exactly \autoref{lem:SimpleRetroPropag} concerning the zero-player case. So we  focus on  the case where $v \in V_{\Max}$.

Let $\Gamma_{\min}^+(v) \subset \Gamma^+(v)$ be the set of $w$ in $\text{argmin}_{w\in\Gamma^+(v)}\ropti{v}{w}$. From \autoref{lem:simpleReturnMin}, we have that, no matter the strategy $\sigma$ chosen on $v$, $h(D(\rho,\sigma)(v)) \in \Gamma_{\min}^+(v)$. From here we distinct two cases. 

\begin{itemize}
    \item Either there is a vertex $w_1 \in \Gamma_{\min}^+(v)$ such that $\val^*(v,w_1)=1$ (there might be several such vertices).  In this case, \autoref{lem:OptimalValRetSimple} (case 2) states that  for all arcs $(u,v)$ with $u \neq w_1$ we have $\val^*(u,v)=1$ and $\ropti{u}{v}=\ropti{v}{w_1}$ if there is an arc $(v,u)$ and $\ropti{u}{v}=1$ if not. 
    If $(w_1,v) \in \arcs$, it remains to consider the case of $(w_1,v)$, which can be done by applying a second time  \autoref{lem:OptimalValRetSimple} (case 2) with $u=w_1$.
    \item If there is no such vertex $w_1 \in \Gamma_{\min}^+(v)$ such that $\val^*(v,w_1)=1$, choose any vertex $w_0 \in \Gamma_{\min}^+(v)$. By using \autoref{lem:OptimalValRetSimple} (case 3) for all arcs $(u,v)$ with $u \neq w_0$ we have that $\val^*(u,v)=0$ and $\ropti{u}{v}=\ropti{v}{w_0}+1$ if there is an arc $(v,u)$, and $\ropti{u}{v}=1$ if not. As for the previous case, if $(w_0,v) \in \arcs$, it remains to consider the case of $(w_0,v)$, which can be done by applying a second time \autoref{lem:OptimalValRetSimple} (case 3) with $u=w_0$.
\end{itemize}

\end{proof}

This allows us to compute the optimal values of all vertices in linear time. 

\begin{theorem}
Given a one-player binary tree-like rotor game on a simple graph, we can compute the optimal values of all vertices in time complexity $O(|V|)$.
\end{theorem}

\begin{proof}
We consider the arcs in the same order than in \autoref{CDA}. For all arcs directed from $u_0$ towards the leaves, we use \autoref{lem:OptimalValRetSimple}. This is done in $O(|\arcs|)$ comparisons. Then, for all arcs directed from the leaves towards $u_0$ we use \autoref{OptiRetroSimple} which is done in $O(|\arcs|)$ comparisons. Then, we just need to compute the optimal value of each vertex as in the proof of \autoref{thm:opt_complex}. Which is done in $O(|{\cal A}|)$ comparisons as well.
\end{proof}

\begin{remark}
This procedure gives us the optimal value of the game simultaneously for every starting vertex, in overall linear time. However, as noted before (see \autoref{simpleStrategy}), the optimal strategy depends on the starting vertex.
\end{remark}

This extension is also valid for the two-player binary variant on a simple graph as the subtree optimal equilibrium is preserved for the same reason than in \autoref{sec:2P-B}.

\subsubsection{One-player Simple Integer Rotor Game}
\label{subsubsec:1PSimpleInt}

This last subsection introduces a notion of \emph{access flow} that measures the potential access to an arc for a rotor walk starting at $u_0$. It allows us to compute the optimal value of the game in linear time for the integer case on a simple graph, in contrast to the general case where we had to use a bisection algorithm (\autoref{subsec:oneP-integ}).

\begin{lemma}[Strategy $\sigma_{max}$]
\label{lem:sigmax}
Given a vertex $u_0$, there is a strategy $\sigma_{max} \in \Sigma_{\Max}$ such that, for all arc $(u_0,v) \in \arcs$,
\[ \rf{u_0}{v}{\sigma_{max}}=\max_{\sigma \in \Sigma_{\Max}} \rf{u_0}{v}{\sigma},\]
namely any strategy that chooses to direct arcs towards $u_0$ when possible.
\end{lemma}

\begin{proof}
We prove recursively that the strategy $\sigma_{max}$ described above is maximal, starting from leaves and going back to $u_0$. Let $(u,v) \in \arcs$ such that $(u,v)$ is directed towards the leaves, suppose that all return flows of arcs $(v,w) \in \arcs$ with $w \in \Gamma^+(v) \setminus \{u\}$ are maximal. 
If $v \in V_r$, $\rf{u}{v}{\sigma}$ is maximal if the return flow of $D(\rho,\sigma)(v)$ is maximal among all strategies $\sigma$. 
By choosing $\sigma_{max}(v)=(v,u)$, we have $B_v(\sigma(v), (v,u), (v,w)) = 1$ for any $w \in \Gamma^+(v) \setminus \{u\}$. At the same time, $\sigma(v)$ is directed towards $u_0$.
From Equation~\eqref{form:returnSimple}, we then have that, no matter the value of the return flow of the arcs $(v,w)$, as long as it is maximal, the return flow of $(u,v)$ is maximal. If there is no arc $(v,u)$, the choice does not matter as $\rf{u}{v}{\sigma}=1$ for any strategy $\sigma \in \Sigma_{\Max}$. 
\end{proof}

\begin{definition}[Access flow $\text{acc}(u,v)$]
\label{def:constraint}
For every arc $(u,v) \in \arcs$,  we denote by  $\text{acc}(u,v)$ the value such that, if we remove all outgoing arcs of $v$ and we add an arc $(v,u)$, $\text{acc}(u,v)$ is the maximal number of times arc $(u,v)$ is visited in the maximal rotor walk starting from $((\rho,\sigma),u_0)$ for all strategies $\sigma \in \Sigma_{\Max}$.
Note that $\text{acc}(u,v)$ might be infinite in the case where all sinks are in $T_{(u,v)}$.
\end{definition}

In particular, $\text{acc}(u,v)$ is positive if and only if there exists a strategy $\sigma$ such that $v$ is visited at least once in the maximal rotor walk starting from $((\rho,\sigma),u_0)$.

A sink $s$ is said to be \emph{reachable} if, for $(u,s) \in \arcs$, $\text{acc}(u,s)$ is positive. Hence, to solve our problem we just need to find the sink $s_{\max}$ of maximal value among the reachable sinks.

From here we present the linear process that allows to compute the access flow of all arcs with head in $S_0$. 

%Please remember that, for a couple of neighbours $(u,v)$, if there is an arc between $u$ and $v$ but no arc between $v$ and $u$ then $\rf{u}{v}{\rho}=1$. 

\begin{lemma}[Access flow around $u_0$]
\label{lem:constraintV}
 Let $\sigma_{\text{max}}$ be the strategy detailed above. For any $v_i \in \Gamma^+(u_0)$ we have:
$$
\text{acc}(u_0,v_i) = \left\{
\begin{array}{ll}
\min_{v_j, j \neq i} \rf{u_0}{v_j}{\sigma_{\text{max}}} - B_{u_0}(\rho(u_0),(u_0,v_j),(u_0,v_i))& \text{ if } u_0 \in V_r \\
\min_{v_j, j \neq i} \rf{v_0}{v_j}{\sigma_{\text{max}}}&\text{ if }u_0 \in V_{\Max}\\
\end{array}
\right.
$$
\end{lemma}

\begin{proof}
To maximize the number of times $(u_0,v_i)$ is visited in the definition of $\text{acc}(u_0,v_i)$, we need to maximize the return flows of all arcs $(u_0,v_j)$. All these arcs are directed from $u_0$ towards  leaves, so the return flow of these arcs under strategy $\sigma_{max}$ is maximum among all strategies. Both cases derive from Equation~\eqref{form:returnSimple}, but in the second one, $u_0$ is initially directed towards $v_i$.
\end{proof}

Similarly to the return flow, we give a recursive equation that computes the access flow on all arcs directed from $u_0$ towards the leaves.

\begin{lemma}[Access flow Propagation]
\label{lem:ConstPropag}
 Let $(u,v) \in \arcs$, let $w_0,..,w_k$ be the vertices of $\Gamma^+(v)\setminus u$. If we know $\text{acc}(u,v)$, then, for any couple $(v,w_i)$ we have: 
 \[\text{acc}(v,w_i)=\min \left\{
 \begin{array}{ll}
    \min_{w_j, j \neq i} \rf{v}{w_j}{\sigma_{max}}-B_{v}(\rho(v),(v,w_j),(v,w_i));\\
    \text{acc}(u,v)- B_{v}(\rho(v),(v,u),(v,w_i))
    \end{array}
    \right.
\]
 If $v \in V_{\Max}$, we have:
 \[\text{acc}(v,w_i)=\min \left\{
 \begin{array}{ll}
    \min_{w_j, j \neq i} \rf{v}{w_j}{\sigma_{max}};\\
    \text{acc}(u,v)
    \end{array}
    \right.
\]
  as the player can always choose $\sigma(v)=(v,w_i)$.
\end{lemma}

Once again, using this formula to compute all access flows might take more than a linear time, so we give a property similar to \autoref{lem:SimpleRetroPropag} allowing to compute all access flows in linear time.

\begin{lemma}[Access flow Computation]
\label{lem:constraintRetro}
Let $(u,v) \in \arcs$ and $w_0,..,w_k$ be the vertices of $\Gamma^+(v)\setminus u$, and assume that  $\text{acc}(u,v)$ is known. Then  $\text{acc}(v,w_i)$ can be determined for all $i \in \{0 \dots k\}$ by computing $\text{acc}(v,w_i)$ for only two values of $i$.
\end{lemma}

\begin{proof}
Same proof than for \autoref{lem:SimpleRetroPropag} if  vertex $v$ belongs to $V_r$. If $v \in V_{\Max}$, we also use \autoref{lem:SimpleRetroPropag} but with $B_{v}(\sigma(v),(v,u),(v,w_i))=0$ for all $w_i \in \Gamma^+(v)\setminus \lbrace u \rbrace$.
\end{proof}

\begin{theorem}
One can compute the value of a one-player tree-like rotor game with arbitrary integer values in linear time $O(|V|)$ for a given starting vertex $u_0$.
\end{theorem}

\begin{proof}
Consider the arcs in the same order than in \autoref{CDA}. For all arcs directed from $u_0$ towards the leaves, we compute the maximal return flows recursively from the leaves (i.e. constructing strategy $\sigma_{max}$). Then, for all arcs with tail $u_0$ we use \autoref{lem:constraintV} and finally we use \autoref{lem:constraintRetro} to compute the access flow on all arcs directed towards the leaves. Each arc is considered at most twice so the complexity of these steps sums up to $O(|\arcs|)$. After that, the value of the game is the reachable sink with maximal value which is computed in $O(|V|)$. As the graph is tree-like and simple we have $|\arcs| \leq 2|V|$ hence our complexity result.
\end{proof}

As a concluding example, let us apply this Theorem to the example of \autoref{simpleIntAlgo} to compute the correct value of the game.

\begin{figure}[!ht]
    \centering
    \begin{tikzpicture}
         
    \node[shape=circle,draw=black] (A) at (0,0) {$u_0$};
    \node[shape=rectangle,draw=black] (B) at (2,0) {$v$};
    \node[shape=circle,draw=black,thick] (C2) at (4,-1) {$0$};
    \node[shape=circle,draw=black] (C) at (4,-1) {$~~~~~$};
    \node[shape=circle,draw=black,thick] (C'2) at (4,1) {$1$};
    \node[shape=circle,draw=black] (C') at (4,1) {$~~~~~$};
    
    \node[shape=circle,draw=black,thick] (E2) at (-2,0) {$2$};
    \node[shape=circle,draw=black] (E) at (-2,0) {$~~~~~$};
    
    \path [->,bend right=20](B) edge (A);
    \path [->,bend right=20, red, dashed, very thick](A) edge (B);
    \path [->](A) edge (E);
    \path [->](B) edge (C);
    \path [->](B) edge (C');

   \draw [ thick,->,>=stealth, red](2.5,2) arc (0:330:0.4cm);

\end{tikzpicture}
\caption{Example of a simple integer tree-like rotor game. Here, $v \in V_{\Max}$ and $u_0 \in V_r$. All other vertices belong to $S_0$ and their value is written inside them.
The initial configuration of $u_0$ is the red arc in dashes.}
\label{simpleIntAlgo}
\end{figure}
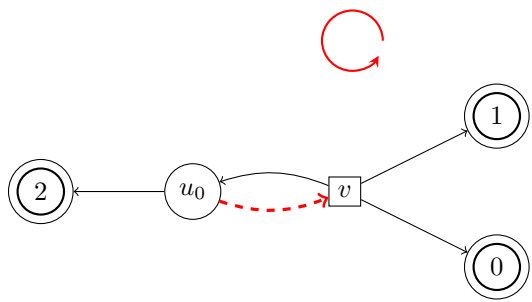

First, compute the return flows of arcs $(u_0,2)$ and $(u_0,v)$ under strategy $\sigma_{\max}$. We have $\rf{u_0}{2}{\sigma_{\max}}=1$ and $\rf{u_0}{v}{\sigma_{\max}}=2$. From \autoref{lem:constraintV}, we have $\text{acc}(u_0,2)=1$ as $B_{u_0}(\rho(u_0),(u_0,v),(u_0,2))=1$. From the same formula we have $\text{acc}(u_0,v)=1$. And by \autoref{lem:constraintRetro} we have $\text{acc}(v,1)=\text{acc}(v,0)=1$. All sinks are reachable, and $2$ is the maximal value among them, hence $\val^*(v_0)=2$.

\section*{Future Work}

Concerning ARRIVAL, one remaining fundamental question is to determine whether there exists a polynomial algorithm to solve the zero-player game. Similarly, problems such as \emph{simple stochastic games, parity games} and \emph{mean-payoff games} are also in NP $\cap$ co-NP, and there are no polynomial algorithm known to solve them (see \cite{halman2007simple}). For those different problems, considering sub-classes of graphs where we can find polynomial algorithms is a fruitful approach (see  \cite{auger2014finding} and \cite{auger2019solving}). This paper is a first step in this direction.

Thus, we would like to study more general classes of graphs. To begin with, even graphs that are well-studied in terms of the  \emph{sandpile group} such that ladders or grids remain now an open problem for ARRIVAL.
The problem of finding the destination of multiple particles at the same time is also an important open problems in nearly all cases except the path graph. 

Finally, another natural extension of our algorithm would be to define and study an adequate notion of graphs with bounded width to generalize the Tree-like multigraph case.

\subsection*{Acknowledgment}

 \addcontentsline{toc}{section}{References}

 \bibliography{bibliArticle}

\end{document}